\documentclass[acmsmall]{acmart-arXiv}

\setcopyright{none}
\renewcommand\footnotetextcopyrightpermission[1]{} 
\AtBeginDocument{%
  \providecommand\BibTeX{{%
    \normalfont B\kern-0.5em{\scshape i\kern-0.25em b}\kern-0.8em\TeX}}}

\acmDOI{10.1145/3393691.3394203}
\acmISBN{978-1-4503-7985-4/20/06}

\usepackage{graphicx}
\usepackage{subcaption}
\usepackage{braket}
\usepackage{bbm}
\usepackage{dsfont}
\usepackage{enumitem}
\usepackage{wrapfig}
\usepackage{sidecap}
\graphicspath{{./figs/}}
\newcommand{\eg}{\emph{e.g.}}
\newcommand{\ie}{\emph{i.e.}}
\newcommand{\vect}[1]{\mathbf{#1}}
\newcommand{\E}{\mathbb E}
\newcommand{\N}{\mathbb N}
\newcommand{\Prob}{\mathbb P}
\newcommand{\ind}{\mathds{1}}
\newcommand{\abs}{\hbox{abs}}
\newtheorem{remark}{Remark}[section]
\allowdisplaybreaks

\begin{document}

\title{On the Analysis of a Multipartite Entanglement Distribution Switch}

\author{Philippe Nain}
\affiliation{%
  \institution{Inria}
  \country{France}}
  \email{philippe.nain@inria.fr}

\author{Gayane Vardoyan}
\affiliation{%
  \institution{University of Massachusetts}
\city{Amherst}
  \country{USA}
}
\email{gvardoyan@cs.umass.edu}

\author{Saikat Guha}
\affiliation{%
  \institution{University of Arizona}
  \city{Tucson}
  \country{USA}}
  \email{saikat@optics.arizona.edu}

\author{Don Towsley}
\affiliation{%
 \institution{University of Massachusetts}
 \city{Amherst}
 \country{USA}}
 \email{towsley@cs.umass.edu}

\renewcommand{\shortauthors}{PHILIPPE NAIN et al.}

\begin{abstract}
 We study a quantum switch that distributes maximally entangled multipartite states to sets of users. The entanglement switching process requires two steps: first, each user attempts to generate bipartite entanglement between itself and the switch; and second, the switch performs local operations and a measurement to create multipartite entanglement for a set of users. In this work, we study a simple variant of this system, wherein the switch has infinite memory and the links that connect the users to the switch are identical. Further, we assume that all quantum states, if generated successfully, have perfect fidelity and that decoherence is negligible. This problem formulation is of interest to several distributed quantum applications, while the technical aspects of this work result in new contributions within queueing theory. 
Via extensive use of Lyapunov functions, we derive necessary and sufficient conditions for the stability of the system and closed-form expressions for the switch capacity and the expected number of qubits in memory.
\end{abstract}



\maketitle

\section{Introduction}
\label{sec:introduction}
The advent of quantum communications brings forth a number of advances that are not accessible via classical communication technology. A striking example is quantum key distribution (QKD) to achieve information-theoretic security, \eg, using well-known quantum-cryptographic protocols such as BB84 \cite{bennett2014quantum} or E91 \cite{ekert1991quantum}. These protocols allow two parties to generate a ``raw'' key, which, after classical error correction on a public authenticated channel (\ie, information reconciliation) and privacy amplification, is distilled into a secret key that can be used as a one-time pad. Some companies, including MagiQ Technologies, Inc. and ID Quantique, offer QKD for commercial and government use; the latter deployed such a system to protect the Geneva state elections in Switzerland in 2007 \cite{marks2007quantum}. 
The original QKD schemes were designed for two trusted parties to be able to share a secret key, but generalizations of these protocols -- which allow $n\geq 3$ trusted parties to securely generate a shared key -- have been proposed, \eg, \cite{grasselli2018finite}, and are known as multiparty QKD.

Point-to-point QKD using prepare and measure schemes, such as single-photon and decoy-state BB84 and CV-QKD protocols suffer from an exponential rate-versus-distance decay \cite{takeoka2014fundamental, pirandola2017fundamental}. The only way to get around this direct-transmission rate-distance tradeoff is to use entanglement-based QKD protocols, where the network delivers entanglement as a pre-shared resource between communicating parties, who then convert it to a quantum-secure shared secret (key) using protocols like BBM92 \cite{bennett1992quantum} and E91. Further, entanglement generation -- if supported by quantum repeaters along the length of the network path connecting the communicating parties -- is not limited by the exponential rate-vs.-distance tradeoff \cite{guha2015rate, pant2017rate, muralidharan2016optimal}. Besides QKD, there are added benefits to the network delivering entanglement as the base resource since entanglement has other uses such as distributed quantum computing \cite{jiang2007distributed,broadbent2009universal} and sensing \cite{komar2014quantum,eldredge2018optimal}, for example. Some of these applications require multiparty entanglement as the enabling resource  -- the main focus of our paper.
It is prudent therefore to model and analyze quantum networks for entanglement distribution, with the goal of deriving guiding principles that may help practitioners make informed decisions toward efficient network and device operation and design. As a step toward this objective, we model and analyze an entanglement distribution switch -- a quantum device that will serve as a fundamental component of quantum networks.

\begin{wrapfigure}{r}{0.5\textwidth}
\centering
\begin{subfigure}[]{0.24\linewidth}
\centering
\includegraphics[width=\linewidth]{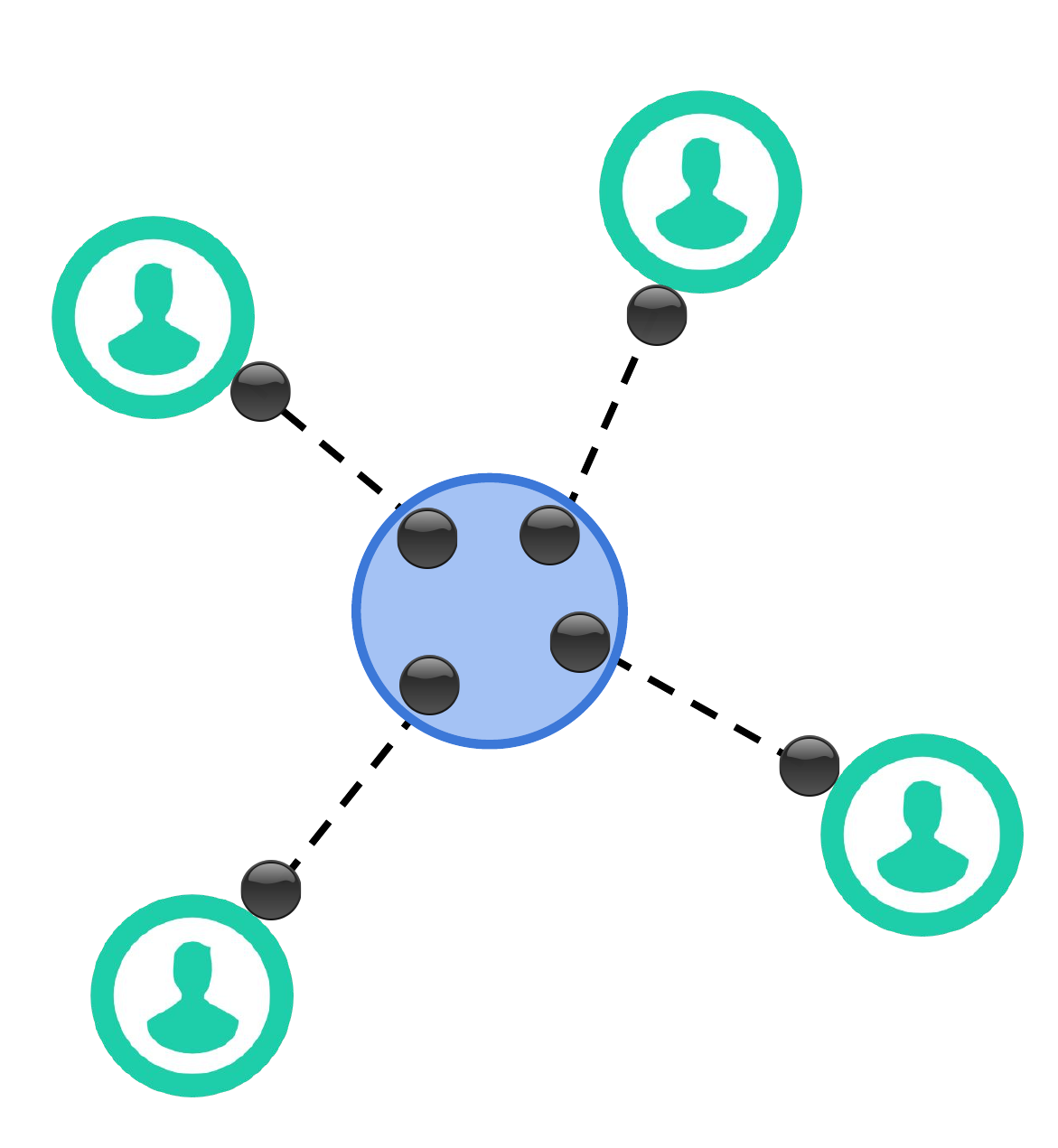}
\caption{no parts}
\label{fig:sc0}
\end{subfigure}
\hfill
\begin{subfigure}[]{0.24\linewidth}
\centering
\includegraphics[width=\linewidth]{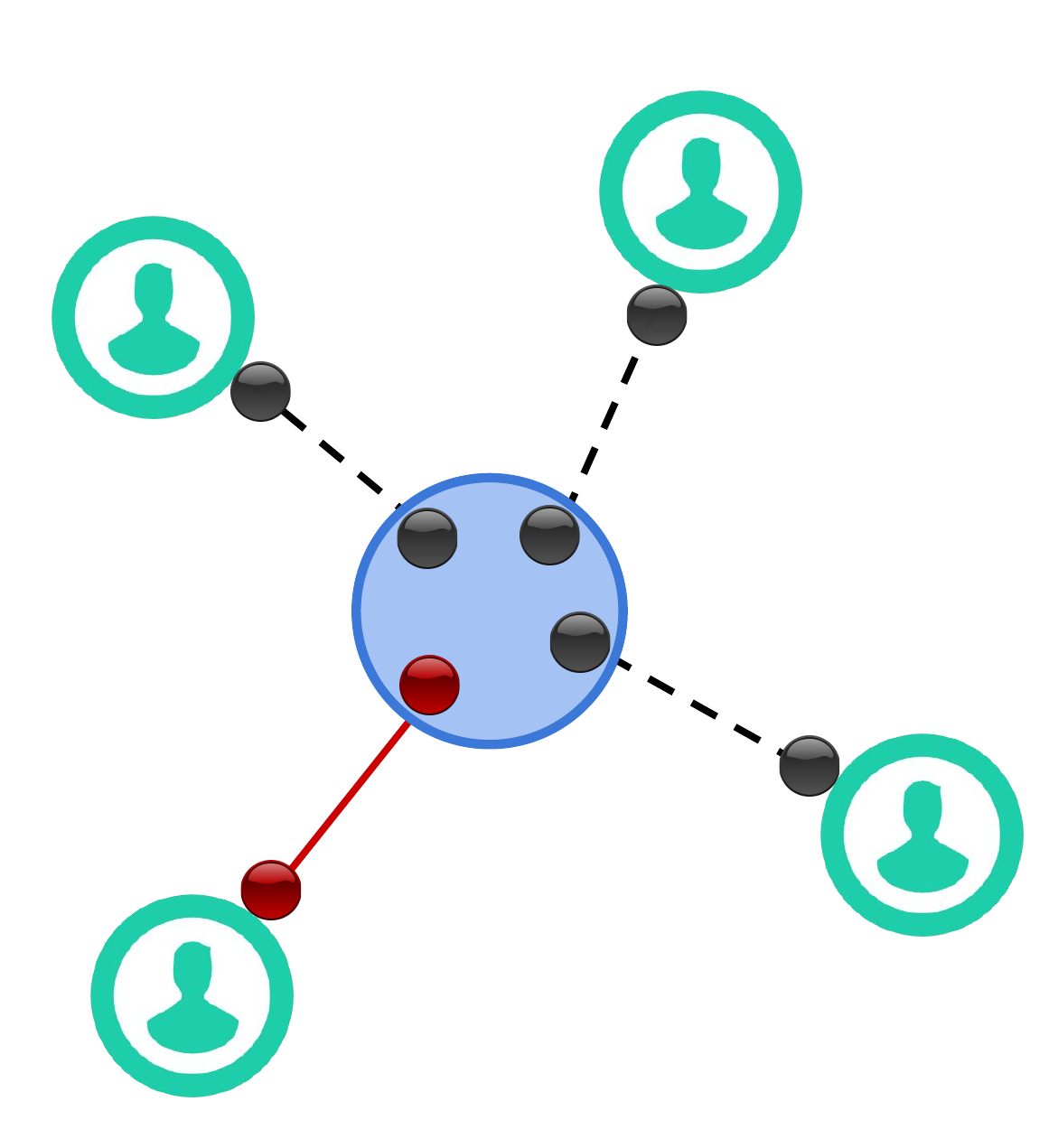}
\caption{one part}
\label{fig:sc1}
\end{subfigure}
\hfill
\begin{subfigure}[]{0.24\linewidth}
\centering
\includegraphics[width=\linewidth]{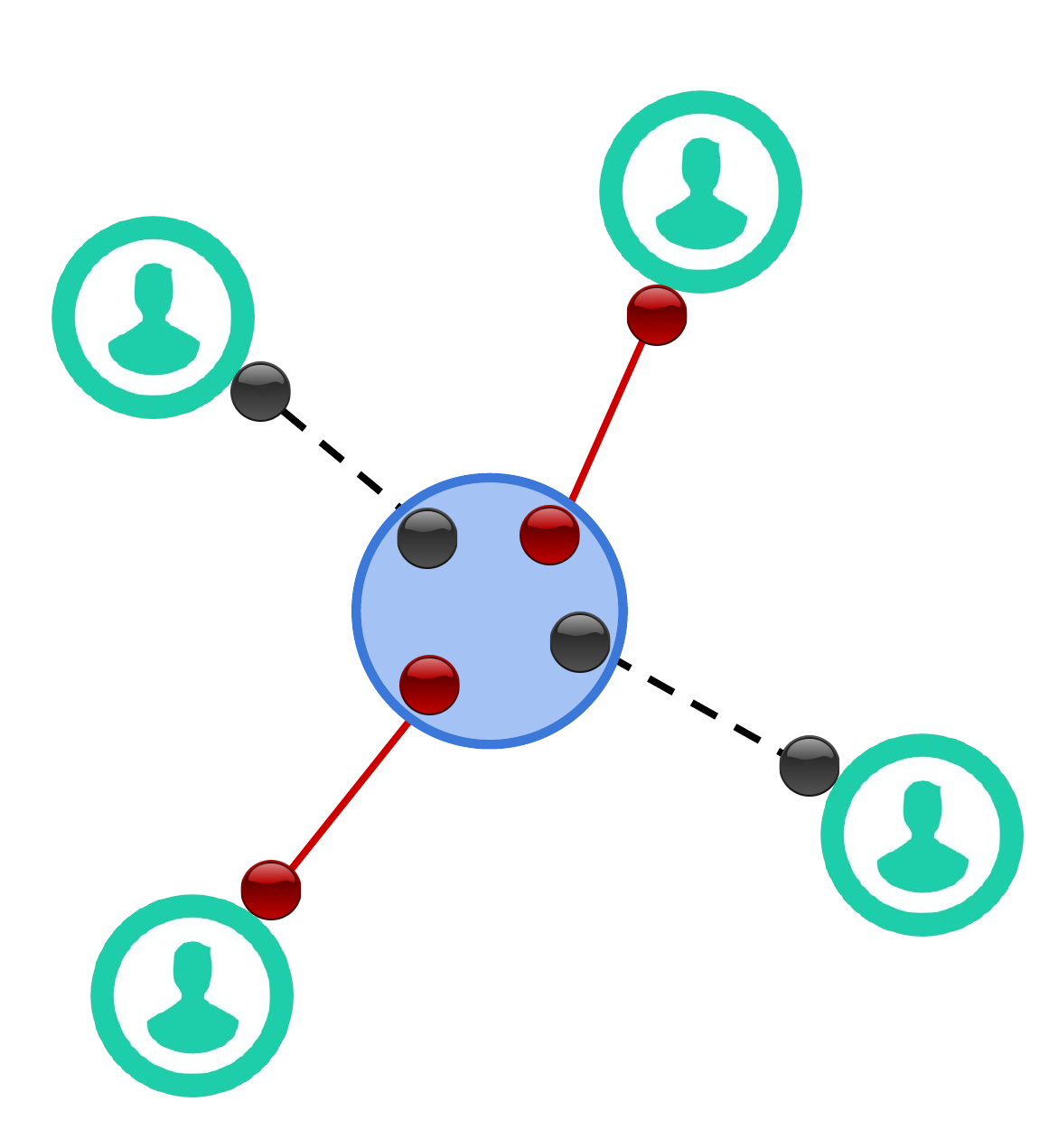}
\caption{two parts}
\label{fig:sc2}
\end{subfigure}
\hfill
\begin{subfigure}[]{0.24\linewidth}
\centering
\includegraphics[width=\linewidth]{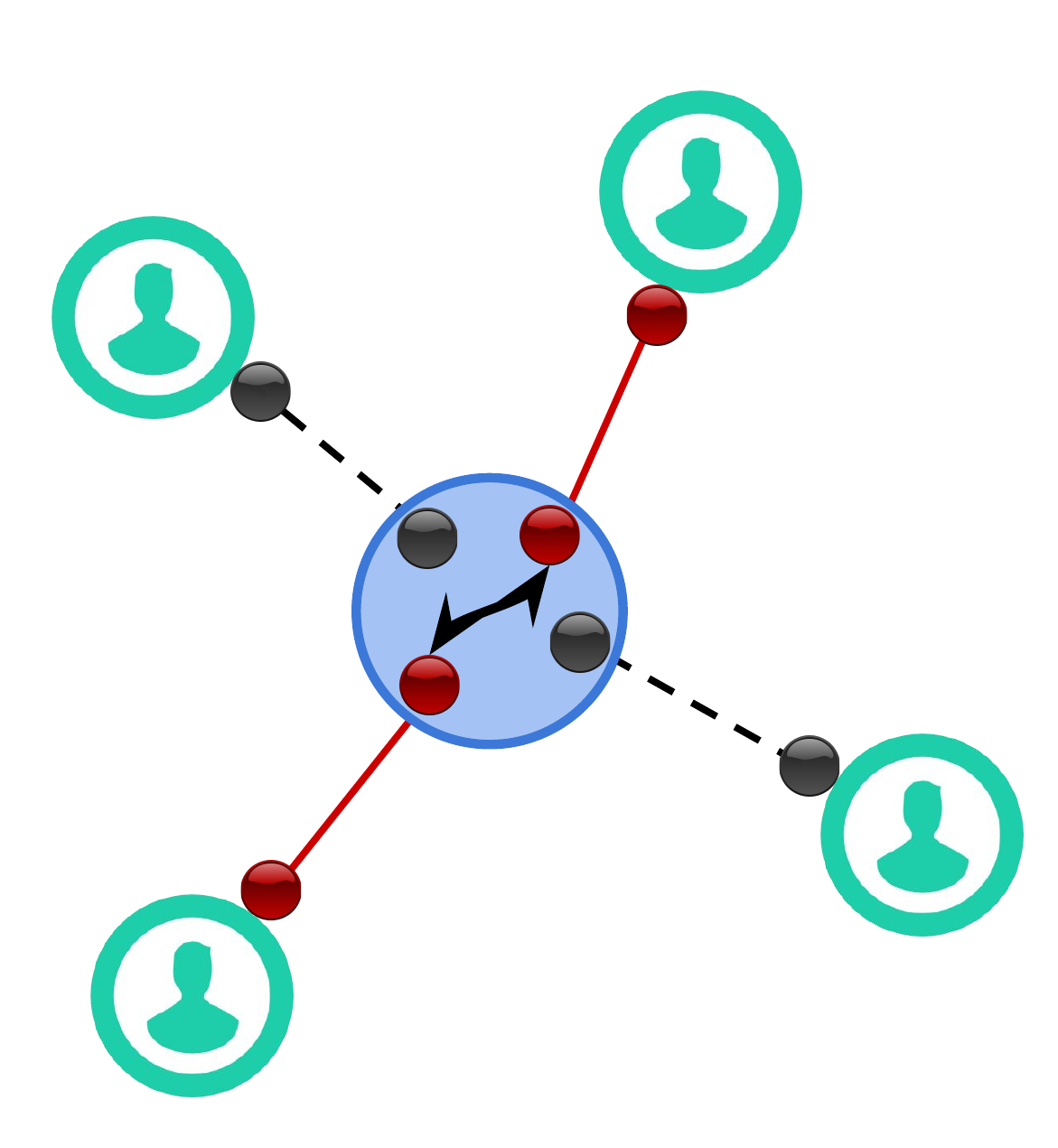}
\caption{assembly}
\label{fig:sc3}
\end{subfigure}
\caption{Each client has a dedicated link to the server. When enough parts arrive, as long as they belong to distinct clients, the server assembles them and they leave the system.}
\label{fig:sc}
\end{wrapfigure}
Before we delve into specifics of quantum switch operation, let us first introduce the problem in abstract, mathematical terms. As we shall see, while the problem considered in this work was initially motivated by its application to entanglement switching, its relevance reaches far beyond this example, and it is of interest to queueing theory in general. Consider a server and $k$ clients in a star topology, as shown in Figure \ref{fig:sc}. Each client has a dedicated link to the server; one may think of these links either as communication channels or physical paths that may be used for object delivery. \emph{E.g.}, in the former scenario, the clients may send data packets to the server for processing, and in the latter, the central node may be an assembly plant that receives components of a product from geographically distant manufacturing plants. Assume that all links are identical, and the server has infinite storage. As soon as $n$ components (or parts) from $n$ distinct clients arrive at the central node, they are processed (or assembled) and immediately leave the system. An example, with $k=4$ and $n=2$, is shown in Panels \ref{fig:sc0}-\ref{fig:sc3}. Assume that the processing/assembly step is instantaneous, but succeeds with probability $q$ (the fact that the assembly step can fail is important in the context of a quantum switch -- see Section \ref{sec:background}). Finally, assume that the arrival process on each link is a Poisson process with a constant rate $\mu$.

Given the last assumption above, we may model the system using a continuous-time Markov chain (CTMC). We will obtain, in closed form, $(i)$ stability conditions, $(ii)$ the maximum achievable assembly rate, $(iii)$ the expected number of items in storage in steady state, when the stationary distribution exists.

When ${n=2}$, the chain is a standard, one-dimensional birth-death process, and one may easily obtain its stationary distribution, provided it is stable. 
The problem becomes significantly more difficult to analyze when $n\geq 3$; 
 in particular, the stability analysis requires the introduction of a Markov chain embedded in the original Markov chain. Via a careful choice of a Lyapunov function and application of the Foster Theorem \cite{foster1953stochastic}, we derive the ergodicity condition of the embedded Markov chain, yielding the ergodicity condition of the original one.
Following is a summary of the results:
\begin{itemize}[leftmargin=*]
\item[-] when the system is stable, the capacity, defined as the maximum achievable number of successful assemblies per time unit, is given by (Section \ref{subsec:capacity})
\begin{align*} 
\frac{q\mu k}{n};
\end{align*}
\item[-] when the system is stable, the expected number of  
items stored at the central node
is given by (Section \ref{subsec:eq})
\begin{align*} 
\frac{k(n-1)}{2(k-n)};
\end{align*}
\item[-] the system is stable if and only if $k>n$ (Section \ref{subsec:stability}).
\end{itemize}

The rest of the paper is organized as follows: in Section \ref{sec:background}, we provide necessary background for quantum communication, its applications, and main challenges. We also describe the operation of  quantum repeaters and switches, and discuss the significance of these devices, as well as of the protocol introduced in this section, to quantum communication. In Section \ref{sec:relwork}, we discuss related work: both from the quantum and the queueing theory perspectives. In Section \ref{sec:analysis}, we introduce the system and construct a CTMC. In Sections \ref{subsec:capacity}, \ref{subsec:eq}, and \ref{subsec:stability}, we analyze the system capacity, expected memory occupancy, and stability, respectively. In Section \ref{subsec:finEQ}, we prove an important property: that when the system is stable, the expected memory occupancy is finite. 
Section \ref{sec:conclusion} concludes the paper.

\section{Background}
\label{sec:background}
A qubit is the quantum analogue of a bit and is represented either by a two-level quantum-mechanical system (\eg, electron spin or photon polarization) or a pseudo two-level system (\eg, the ground and first excited states of an atom). Two or more qubits are said to be entangled when the state of one cannot be described independently from the state of the other(s). In this work, all quantum states, if generated successfully, are assumed to have unit fidelity\footnote{Fidelity $F(\rho,\sigma)\in[0,1]$ is a measure of ``closeness'' between two quantum systems with density operators $\rho$ and $\sigma$. Unit fidelity is achieved when $\rho=\sigma$.} to the corresponding ideal pure state.

Distributed quantum tasks, including those that require two or more parties to share an entangled state, typically require qubit transmission between the participants. 
For optical fiber, channel transmissivity is given by $\eta =e^{-\alpha L}$, where $L$ is the length of the link and $\alpha$ the fiber attenuation coefficient. Transmissivity can be interpreted as the fraction of transmitted power that reaches the end of the fiber; for a single photon, this corresponds to the probability that the photon successfully reaches the receiver\footnote{Transmission through free space poses its own challenges, such a photon loss and phase changes due to scattering \cite{Van_Meter2014-vz}.}. In other words, the likelihood of losing a quantum state in transit grows exponentially with distance. At the same time, the no-cloning theorem \cite{wootters1982single} prevents one from making an independent copy of an unknown state, thereby rendering losses irrecoverable. Hence, safe transmission of a quantum state across a large distance is a major challenge in quantum communication. 

One method of overcoming the limited-distance problem  is through the use of quantum repeaters \cite{briegel1998quantum} coupled with the process of teleportation \cite{bennett1993teleporting}.
For example, consider a quantum protocol that requires two parties, $A$ and $B$, to share a two-qubit entangled state. A quantum repeater, $R$, is positioned between the two users, and as a first step the device attempts to generate entangled links between itself and the users; we call these ``link-level entanglements''. Assume that there is at least one entangled photon pair source (which may be colocated with $R$) that assists with this task, by generating Bell pairs\footnote{Bell pairs are maximally-entangled two-qubit states, \eg, $(\ket{00}+\ket{11})/\sqrt{2}$.} and distributing one half of a pair to a user and the other half to $R$. Once both entangled pairs are successfully distributed, the repeater performs an entanglement swapping operation \cite{zukowski1993event}, as a result of which $A$'s and $B$'s qubits become entangled (despite not having directly interacted anytime in the past); we call this new entanglement an ``end-to-end entanglement'' between $A$ and $B$. The process of extending the range of entanglement using quantum repeaters is what makes quantum distributed applications, such as QKD, feasible over large distances.

Entanglement swapping can be generalized to create multipartite end-to-end entangled states between three or more users \cite{zukowski1995entangling,bruss1997quantum,bose1998multiparticle}, \eg, using a collection of Bell pairs.
In this work, we consider the generation of maximally-entangled Greenberger-Horne-Zeilinger (GHZ) states of $n$ qubits \cite{greenberger1989going}, \eg, states of the form
${(\ket{0}^{\otimes n}+\ket{1}^{\otimes n})/\sqrt{2}}$, where ${n\geq 3}$. These states have a wide variety of uses in distributed quantum computation, information, and communication; examples include quantum secret sharing \cite{hillery1999quantum,xiao2004efficient,tittel2001experimental}, multiparty QKD (also known as NQKD or cryptographic conferencing) \cite{chen2004multi,epping2017multi,grasselli2018finite}, quantum sensing \cite{eldredge2018optimal,schaffry2010quantum}, and multipartite generalization of superdense coding \cite{bose1998multiparticle,hao2001controlled}. Entanglement swapping protocols typically involve two-qubit gates and projective measurements (\eg, Bell-state measurements or GHZ state projections), and certain operations succeed only probabilistically, which is especially true in linear optics \cite{ewert20143,grice2011arbitrarily,helmer2009measurement}. To account for this phenomenon in our analysis, we introduce a parameter $q$ which corresponds to the successful entanglement swapping probability (analogous to the ``assembly'' success probability discussed in Section \ref{sec:introduction}).

Throughout this work, we use the term ``quantum switch'' as opposed to ``quantum repeater'' because the former will be equipped with entanglement switching logic. We also assume that our quantum switch will be equipped with an infinite number of quantum memories. All quantum memories are assumed to have infinite coherence times. Suppose that the switch serves $n$-partite entangled states to $k\geq n$ users in a star topology, with each user having a dedicated link connecting it to the switch. Each link continuously attempts to generate maximally entangled bipartite (Bell) states between the corresponding user and the switch. Qubits from newly-generated Bell pairs are stored: one qubit at the user and the other at the switch. When $n$ Bell pairs are successfully generated across $n$ distinct links, the switch performs entanglement swapping to create an $n$-partite GHZ state. We would like to obtain closed-form expressions for the expected number of stored qubits at the switch as well as the capacity of the switch, expressed in the maximum possible number of swapping operations performed per unit time, or equivalently, the maximum rate of entanglement switching. 
To determine the latter, we must assume that any combination of $n$ users wish to share an entangled state, as to allow the switch to combine any $n$ Bell pairs, as long as they all belong to distinct users, without any restrictions. 
In fact, removing this assumption would necessarily lower the rate at which the switch serves end-to-end entangled states, and the final expression would no longer correspond to its \emph{capacity}.
Further, since the switch has infinite memory, conditions for stability of this system are also of interest. Finally, we assume that entanglement generation at the link level is a Poisson process, with $\mu$ being the successful entanglement generation rate. Under these assumptions, it becomes clear that the entanglement switching policy described here is mathematically equivalent to the assembly-like process described in Section \ref{sec:introduction}.
\section{Related Work}
\label{sec:relwork}
The problem introduced in this paper can be viewed as a type of stochastic ``assembly-like queue'' \cite{harrison1973assembly,hopp1989bounds,bonomi1987approximate,lipper1986assembly,bhat1986finite} or a ``kitting process'' \cite{som1994kitting,de2011performance,ramachandran2005performance}. The aforementioned references focus on problems that are somewhat similar to our problem, yet differ significantly enough -- in formulation, analysis, or both -- that the latter warrants an independent analysis. Specifically, in some of these works there are assumptions (\eg, finite buffer capacity, limited and fixed number of input streams, restrictions on the assembly process or service time distribution, etc.) that prohibit us from extending their analyses to our problem, especially when it comes to the question of stability for our system, which has unlimited buffer. Meanwhile, other studies present only approximate analyses or bounds. In contrast, our goal is to obtain closed-form expressions for throughput and expected number of items in storage at the central node.

Recall that we choose a CTMC to model our system. In the next section, we will see that this Markov chain 
has a countably infinite multidimensional state space due to the limitless buffer criterion.
The ergodicity of multidimensional Markov chains, even those with simple structure and transition probabilities, is often difficult to prove due to a lack of analytical tools in this domain. 
An exception concerns a family of two-dimensional random walks, for which necessary and sufficient stability conditions exist \cite[Theorem 1.2.1]{fayolle1999random}, \cite{FaMaMe-1995}.
One of the most well-known and powerful tools to determine whether an irreducible Markov chain is positive recurrent is the Foster-Lyapunov Criterion. Unfortunately, this criterion cannot directly be applied to our Markov chain due to the requirement that the drift must be negative on all but a finite subset of the state space. 
We encounter similar problems when attempting to apply extensions of Foster's Theorem to higher-dimensional chains, \eg, as in \cite{szpankowski1988stability,rosberg1980positive,kompalli2009generalized}.
For many multidimensional Markov chains, including ours, there are an infinite number of ``boundary'' states for which the drift is positive, regardless of the choice of Lyapunov function. Because of this, we 
consolidate a number of different
techniques to obtain necessary and sufficient conditions for our chain's ergodicity.

In \cite{vardoyan2019capacityarXiv}, we analyze a quantum switch that serves \emph{either} bipartite \emph{or} tripartite entangled states to sets of users. The analyses are restricted to the identical-link case and buffer sizes of one and two per link. A set of randomized switching policies are explored, together comprising an achievable capacity region. In \cite{vardoyan2019stochasticarXiv}, we use Markov chains to analyze an entanglement switch that serves only bipartite end-to-end entangled states to users by performing Bell-state measurements. In addition to swapping failures, the model also accounts for finite buffer sizes, potentially non-identical links (implying different entanglement generation rates for different links), and incorporates the effects of state decoherence. While these extensions to the identical-link, infinite-buffer, no-decoherence variant of the problem introduce complexity into the models, they do not preclude us from deriving closed-form expressions for the switch capacity, the expected number of stored qubits, and even the stationary distribution, when it exists. In general, the bipartite switching case is significantly easier to analyze using CTMCs than even the simplest variant of the $n$-partite switching system, for $n\geq 3$. Nevertheless, the results obtained in this work serve a useful purpose: the capacity of the $n$-partite quantum switch derived here is an upper-bound to that of a more realistic system, where the links are non-identical, quantum memory is finite, and where quantum states are susceptible to decoherence. It is easy to see why the latter two properties would decrease the capacity of the system, \ie, restricted storage space and expiration of resources  cannot increase capacity. For the case of non-identical links, one may derive an upper bound on the capacity by transforming the system into an identical-link one, where all links generate entanglement at rate $\mu_{\max}$ -- the rate at which the fastest link in the original system successfully generates entanglement. This yields an upper bound on the capacity given by $qk\mu_{\max}/n$.
\section{The Model \& Preliminaries}
\label{sec:analysis}
\subsection{The Model}
In this section, we formally introduce the model and construct a CTMC that serves as the basis of all our analysis. First,
some notes on notation: $\N$ is the set of nonnegative integers and $\mathbb{R}$ is the set of reals.
In this paper, boldface is reserved for vectors (and sometimes, matrices).
$\vect{0}$ is the row vector of dimension $n-1$ with all-zero entries, $\vect{1}$ is the row vector of dimension $n-1$ that has all entries equal to $1$, and $\vect{e}_i$ is the row vector of dimension $n-1$ that has all entries equal to $0$ except the $i$th entry that is equal to $1$. $x_i$ refers to the $i$th entry of vector $\vect{x}$.  
The expression $\vect{x}>\vect{y}$ refers to an element-wise comparison, \ie, $x_i>y_i$, $\forall ~i$. In this paper, all vectors are nonnegative so that for a $J$-dimensional vector $\vect{x}$, $|\vect{x}|\equiv x_1+\dots+x_{J}$.
$\ind_A$ denotes the indicator function of the event $A$, with $\ind_A=1$ if $A$ is true and $\ind_A=0$ otherwise. Throughout the analysis, $3\leq n\leq k$.

Recall our assumption that any combination of $n$ users wish to communicate, \ie, the switch performs entanglement swapping as soon as there are $n$ entangled Bell pairs belonging to $n$ distinct links. This, coupled with the observation that all links are identical in that they all generate Bell pairs at the same rate of $\mu$, allows us to represent the state of the system with an $(n-1)$-dimensional vector $\vect{x}=(x_1~ x_2~ \dots~x_{n-1})$ indicating the number of available entangled pairs (or equivalently, the number of stored qubits at the switch) belonging to different links. Note that it is possible for some entries in $\vect{x}$ to equal 0, in cases where fewer than $n-1$ links have Bell pairs. Thus, we may construct a CTMC with state space $R\coloneqq \{(x_1~\dots ~x_{n-1}): x_1\geq 0,\dots,x_{n-1}\geq 0 \}$, which can be partitioned as follows:
\begin{align}
\label{eq:partition}
R = R_0 \cup R_1\cup \cdots \cup R_{n-1},
\end{align}
where $R_j$ contains the set of states $\vect{x}$ such that $j$ entries are $0$ and ${n-1-j}$ entries are $\geq 1$. Note that $R_0 = \{(x_1~\dots~x_{n-1}): x_1\geq 1,\dots,x_{n-1}\geq1\}$ and $R_{n-1}=\vect{0}$.

Following, we describe the non-zero rates for this chain. First, note that for any state $\vect{x}$ in $R_0$, there are $n-1$ links with available Bell pairs. Once the $n$th link generates a Bell pair, the switch performs a swap and the chain transitions to state $\vect{x}-\vect{1}$, regardless of whether the swap was successful or not. This transition occurs with rate $(k-(n-1))\mu$. Next, for any state $\vect{x}\in R\setminus R_{n-1}$, a link $l$ with already-stored Bell pairs may generate another one. Then, the chain transitions to state $\vect{x}+\vect{e}_l$, with rate $\mu$.
On the other hand, a state $\vect{x}\in R_j$, for $j\in\{1,\dots,n-1\}$ may generate a Bell pair on a link without any available pairs, and since there are $k-(n-1-j)$ links with no stored Bell pairs, this event occurs with rate $(k-(n-1-j))\mu$. However, we must divide this aggregate transition rate by an appropriate quantity to account for symmetries within the Markov chain. To illustrate this point, consider state $\vect{0}$: when a link without entangled pairs generates a Bell pair -- an event that occurs with rate $k\mu$ in this state -- the chain transitions to a state with all-zero entries, except for one entry that equals one. When the chain is stable, all of these states have equal steady-state probabilities. There are $n-1$ such states $\vect{e}_l$, $l\in\{1,\dots,n-1\}$, so the individual transition rate from $\vect{0}$ to one of these states is $k\mu/(n-1)$. In general, the chain transitions from $\vect{x}\in R_j$ to state $\vect{x}+\vect{e}_l$, where $l$ is such that $x_l=0$, with rate $(k-(n-1-j))\mu/j$.

\subsection{Preliminaries}
Our objective is to derive closed-form expressions for the switch capacity as defined in Section \ref{sec:background} and the expected number of stored qubits at the switch in steady state. To derive these performance metrics, we will first uniformize the CTMC introduced above, by sampling it according to a Poisson process with constant rate $k\mu$. This yields a discrete-time Markov Chain (DTMC) to which we will apply Foster's Theorem. Before we introduce the DTMC, we present Foster's Theorem \cite{Bremaud99} below for ease of reference.
\begin{theorem}[Foster's Theorem]
Let the transition matrix $\vect{P}$ on the countable state space $E$ be irreducible and suppose that there exists a function $h: E \to \mathbb{R}$ such that $\inf_i h(i)> -\infty$ and
\begin{align*}
&\sum\limits_{k\in E} p_{ik}h(k) < \infty \text{ for all }i\in F,\\
&\sum\limits_{k\in E}p_{ik}h(k) \leq h(i)-\epsilon \text{ for all }i\notin F,
\end{align*}
for some finite set $F$ and some $\epsilon >0$. Then the corresponding homogeneous Markov chain is positive recurrent.
\end{theorem} 

Let us define the DTMC that results from the uniformization of the CTMC by $X:=\{\vect{X}_{t}:=(X_{t}^1,\ldots,X_{t}^{n-1}), t=1,2,\ldots\}$, where $X_t^l$ is the number of stored qubits at the switch for link $l$ at time $t$ (or equivalently, the number of available Bell pairs for link $l$ at time $t$). The non-zero transition probabilities for the DTMC are
\begin{align}
\vect{x} \in R_0 &\to \begin{cases}
\vect{x}-\vect{1}, &\text{with prob.} ~\frac{k-(n-1)}{k},\\
\vect{x}+\vect{e}_l, &\text{with prob.} ~\frac{1}{k},~l=1,\dots,n-1,
\end{cases}\label{eq:trans1}\\
\vect{x}\in R_j &\to \begin{cases}
\vect{x}+\vect{e}_l, &\text{with prob.}~\frac{k-(n-1-j)}{kj}~ \text{if }x_l=0,\\
\vect{x}+\vect{e}_l,&\text{with prob.}~\frac{1}{k}~ \text{if }x_l\geq 1,
\end{cases}\label{eq:trans2}\\
&\qquad\text{for}~ j=1,\dots,n-2,~l=1,\dots,n-1,\text{ and}\nonumber\\
\vect{0} &\to \vect{e}_l ~\text{with prob.}~\frac{1}{n-1},~l=1,\dots,n-1\label{eq:trans3}.
\end{align}
Note that $X$ is irreducible and aperiodic. When in addition $X$ is positive recurrent, \ie, stable, we define $\vect{Q}$ as the stationary instance of the vector $\vect{X}_t$.
The steady state of the initial CTMC is the same as that of the sampled CTMC (\ie, the DTMC) thanks to the PASTA property.

Let $C$ be the switch capacity 
and 
$\E[|\vect{Q}|]$ the expected number of stored qubits across all $k$ links in steady state.
Recall from Section \ref{sec:background} that entanglement swapping may fail and that we use $q$ to denote the success probability. Note that this phenomenon affects $C$ but not $\E[|\vect{Q}|]$ or the chain's stability. In  particular, if the CTMC has a stationary distribution given by $\pi$, then
\begin{align*}
C &= q\mu (k-(n-1))\sum\limits_{\vect{x}\in R_0}\pi(\vect{x}),\qquad\text{and}\qquad
\E[|\vect{Q}|] = \sum\limits_{\vect{x}\in R}|\vect{x}|\pi(\vect{x}).
\end{align*}
Note that all entry permutations of $\vect{x}$ are equiprobable. This fact will be useful in Section \ref{subsec:stability}, when we derive the stability condition for the system.
Throughout the derivations of $C$ and $\E[|\vect{Q}|]$, we will often use the following lemma:
\begin{lemma}
\label{lemma:eq40}
If $X$ is stable and $\pi$ is its stationary distribution, then for every mapping $V:\{0,1,\ldots\}^{n-1}\to [0,\infty)$ such that $\E[V(\vect{X}_1)]<\infty$,
\begin{align}
0&=\sum_{\vect{x}\in R_0}  \pi(\vect{x})\Biggl[ \left(\frac{k-(n-1)}{k}\right) (V(\vect{x}-\vect{1})-V(\vect{x}))
+ \sum_{l=1}^{n-1}\frac{1}{k} (V(\vect{x}+\vect{e}_l)-V(\vect{x}))\Biggr]\nonumber\\
&\quad+\sum_{j=1}^{n-2} \sum_{\vect{x}\in R_j} \pi(\vect{x})  \Biggl[ \sum_{l:x_l=0}\left(\frac{k-(n-1-j)}{kj}\right)  (V(\vect{x}+\vect{e}_l)-V(\vect{x}))
+ \sum_{l: x_l\geq 1} \frac{1}{k}  (V(\vect{x}+\vect{e}_l)-V(\vect{x})) \Biggr]\nonumber\\
&\quad+ \pi(\vect{0})\sum_{l=1}^{n-1} \frac{1}{n-1}  (V(\vect{e}_l)-V(\vect{0})).
\label{eq:40}
\end{align}
\end{lemma}
\begin{proof}
Assume that the DTMC $X$ is stable, with $\pi$ its stationary distribution. Further assume that it is in steady state at time $t=1$ (which implies that it is in steady state at any time $t>1$). 
For every mapping $V:\{0,1,\ldots\}^{n-1}\to [0,\infty)$ such that $\E[V(\vect{X}_1)]<\infty$, we have
\begin{align*}
0&=\E[V(\vect{X}_{t+1})-V(\vect{X}_t)]
=\sum_{\vect{x}\in R} \pi(\vect{x}) \E[V(\vect{X}_{t+1})-V(\vect{x})\,|\, \vect{X}_{t}=\vect{x}],
\end{align*}
which immediately yields Eq. (\ref{eq:40}), by using transition probabilities in  (\ref{eq:partition})-(\ref{eq:trans3}).
\end{proof}
\section{Capacity}
\label{subsec:capacity}
Here, we derive the switch capacity $C$, defined as the maximum number of $n$-partite entangled states generated per time unit, or equivalently, the number of successful entanglement swapping operations performed by the switch per time unit.
\begin{proposition}[Capacity]
\label{capacity-n}
If $X$ is stable\footnote{The system is stable if and only if $k>n$, see Section \ref{subsec:stability}.\label{foot:stab}} 
then 
\[
C=\frac{q\mu k}{n}.
\]
\end{proposition}
\begin{proof}
By Lemma \ref{lemma:eq40}, Eq. (\ref{eq:40}) holds.
Take $V(\vect{x})=x_1+\cdots+x_{n-1}$.  
It is shown in Proposition \ref{prop:expectationX} that $\E[|\vect{Q}|]<\infty$ when the system is stable. Since $\E[V(\vect{X}_1)]=\E[|\vect{Q}|]$, we see that the condition $\E[V(\vect{X}_1)]<\infty$ is met so that $(\ref{eq:40})$ holds for this choice of $V$.
After this substitution and multiplication of both sides of (\ref{eq:40}) by $k$, we obtain
\begin{align}
-(n-1)(k-n) \sum_{\vect{x}\in R_0}  \pi(\vect{x}) + k \sum_{j=1}^{n-1} \sum_{\vect{x}\in R_j}\pi(\vect{x})= 0.
\label{eq:41}
\end{align}
From the identities
\[
1=\sum_{\vect{x}\in R}\pi(\vect{x})=\sum_{\vect{x}\in R_0} \pi(\vect{x})+\sum_{j=1}^{n-1}\sum_{\vect{x}\in R_j}\pi(\vect{x}),
\]
where the latter identity holds from  (\ref{eq:partition}), we deduce that
\begin{align}
\sum_{j=1}^{n-1}\sum_{\vect{x}\in R_j} \pi(\vect{x})= 1-\sum_{\vect{x}\in R_0} \pi(\vect{x}).
\label{eq:42}
\end{align}
Hence,  cf. (\ref{eq:41}), (\ref{eq:42}),
\[
0=-n(k-(n-1)) \sum_{\vect{x}\in R_0}  \pi(\vect{x}) +  k,
\]
so that
\begin{align}
\sum_{\vect{x}\in R_0}  \pi(\vect{x})=\frac{k}{n(k-(n-1))}.
\label{eq:43}
\end{align}
To compute the capacity of the switch, observe that entanglement swapping occurs whenever there are $n-1$ distinct links with at least one Bell pair each (\ie, the system is in a state $\vect{x}\in R_0$), and a link without available Bell pairs successfully generates one. This occurs with rate $(k-(n-1))\mu$. Further, an $n$-partite entangled state is generated when the swapping operation succeeds, which occurs with probability $q$.
The capacity $C$ is then given by 
\[
C=q\mu (k-(n-1))\times \sum_{\vect{x}\in R_0} \pi(\vect{x})= \frac{q\mu k}{n},
\]
which proves the proposition.
\end{proof}
\section{Expected Number of Qubits in Memory at Switch}
\label{subsec:eq}
Recall that $\vect{Q}$ is the steady-state version of the vector $\vect{X}_t$. In this section, we derive the expected number of stored qubits, across all links, at the switch in steady state, or $\E[|\vect{Q}|]$. 
\if{false}
\begin{proposition}[Expected Number of Stored Qubits]\hfill\\
\label{prop:eqShort}
If the system is stable\footnote{The system is stable if and only if $k>n$, see Section \ref{subsec:stability}.} and $\E[|\vect{Q}|^2] <\infty$,  then
\[
\E[|\vect{Q}|]=\frac{k(n-1)}{2(k-n)}.
\]
\end{proposition}
\begin{proof}
The beginning of this proof is identical to that of the capacity proof of Section \ref{subsec:capacity}, up to and including Eq. (\ref{eq:40}). Next, let us first take $V(\vect{x})=x_1^2+\cdots+x_{n-1}^2$.
After substituting $V(\vect{x})$ into (\ref{eq:40}) and multiplying both sides by $k$, we obtain
\begin{align}
0&= \sum_{\vect{x}\in R_0}  \pi(\vect{x})\Biggl[(k-(n-1))(-2|\vect{x}|+(n-1) )+2|\vect{x}|+n-1\Biggr]\nonumber \\
&+ \sum_{j=1}^{n-2} \sum_{\vect{x}\in R_j}\pi(\vect{x})\Biggl[k-(n-1-j)
+2|\vect{x}|+n-1-j\Biggr]
+ k\pi(\vect{0})\nonumber\\
&=-2(k-n)\sum_{\vect{x}\in R_0}  \pi(\vect{x})|\vect{x}|+(k-n+2)(n-1)\sum_{\vect{x}\in R_0}  \pi(\vect{x})\nonumber\\
&+ 2\sum_{j=1}^{n-2} \sum_{\vect{x}\in R_j}\pi(\vect{x})|\vect{x}| +  k\sum_{j=1}^{n-1} \sum_{\vect{x}\in R_j}\pi(\vect{x}).
\label{eq:70}
\end{align}
Using Eqs (\ref{eq:42}) and (\ref{eq:43}), we deduce that
\begin{align}
\label{eq:72}
\sum_{j=1}^{n-1} \sum_{\vect{x}\in R_j}\pi(\vect{x})= 1-\frac{k}{n(k-(n-1))}=
\frac{(k-n)(n-1)}{n(k-(n-1))}.
\end{align}
Hence, from (\ref{eq:43}), (\ref{eq:70}), and (\ref{eq:72}),
\begin{equation}
\label{eq:73}
(k-n)\sum_{\vect{x}\in R_0}  \pi(\vect{x})|\vect{x}|
- \sum_{j=1}^{n-2} \sum_{\vect{x}\in R_j}\pi(\vect{x})|\vect{x}|=\frac{k(n-1)}{n}.
\end{equation}
Now, take $V(\vect{x})=(x_1+\cdots+x_{n-1})^2$.  After substituting into (\ref{eq:40}) and multiplying both sides by $k$, we obtain
\begin{align*}
&\sum_{\vect{x}\in  R_0}  \pi(\vect{x})\Biggl[(k-(n-1))(-2(n-1)|\vect{x}|+(n-1)^2 )+(n-1)(2|\vect{x}|+1)\Biggr]\\
&+ \sum_{j=1}^{n-2} \sum_{\vect{x}\in R_j}\pi(\vect{x})\Biggl[(k-(n-1-j))(2|\vect{x}|+1)
+(n-1-j)(2|\vect{x}|+1)\Biggr]\\
&+ k\pi(\vect{0})\\
&=-2(k-n)(n-1)\sum_{\vect{x}\in R_0}  \pi(\vect{x})|\vect{x}|
+2k\sum_{j=1}^{n-2} \sum_{\vect{x}\in R_j} \pi(\vect{x})|\vect{x}|\\
&+(n-1)((n-1)(k-(n-1))+1) \sum_{\vect{x}\in R_0}  \pi(\vect{x})
+ k\sum_{j=1}^{n-1} \sum_{\vect{x}\in R_j}\pi(\vect{x}) = 0.
\end{align*}
Using (\ref{eq:43}) and  (\ref{eq:72}) yields
\begin{align}
2(k-n)(n-1)\sum_{\vect{x}\in R_0}  \pi(\vect{x})|\vect{x}| -2k\sum_{j=1}^{n-2} \sum_{\vect{x}\in R_j} \pi(\vect{x})|\vect{x}|   =k(n-1).
\label{eq:74}
\end{align}
Let $A\coloneqq \sum_{\vect{x}\in R_0}  \pi(\vect{x})|\vect{x}|$ and $B\coloneqq\sum_{j=1}^{n-2} \sum_{\vect{x}\in R_j} \pi(\vect{x})|\vect{x}|$.
Note that $\E[|\vect{Q}|]\coloneqq A+B$.
Eqns (\ref{eq:73})-(\ref{eq:74}) define the following set of linear equations in the unknowns $A$ and $B$:
\begin{eqnarray*}
(k-n)A- B&=&\frac{k(n-1)}{n}\\
2(k-n)(n-1)A-2kB&=&k(n-1).
\end{eqnarray*}
When $k\not=n$ there is a unique solution given by 
\begin{eqnarray*}
A&=&\frac{k(n-1)(2k-n)}{2n(k-n)(k-(n-1))}\\
B&=& \frac{k(n-1)(n-2)}{2n(k-(n-1))}. 
\end{eqnarray*}
We find $\E[|\vect{Q}|]=A+B=\frac{k(n-1)}{2(k-n)}$.
\end{proof}
\fi
\begin{figure}
\centering
\includegraphics[width=0.48\textwidth]{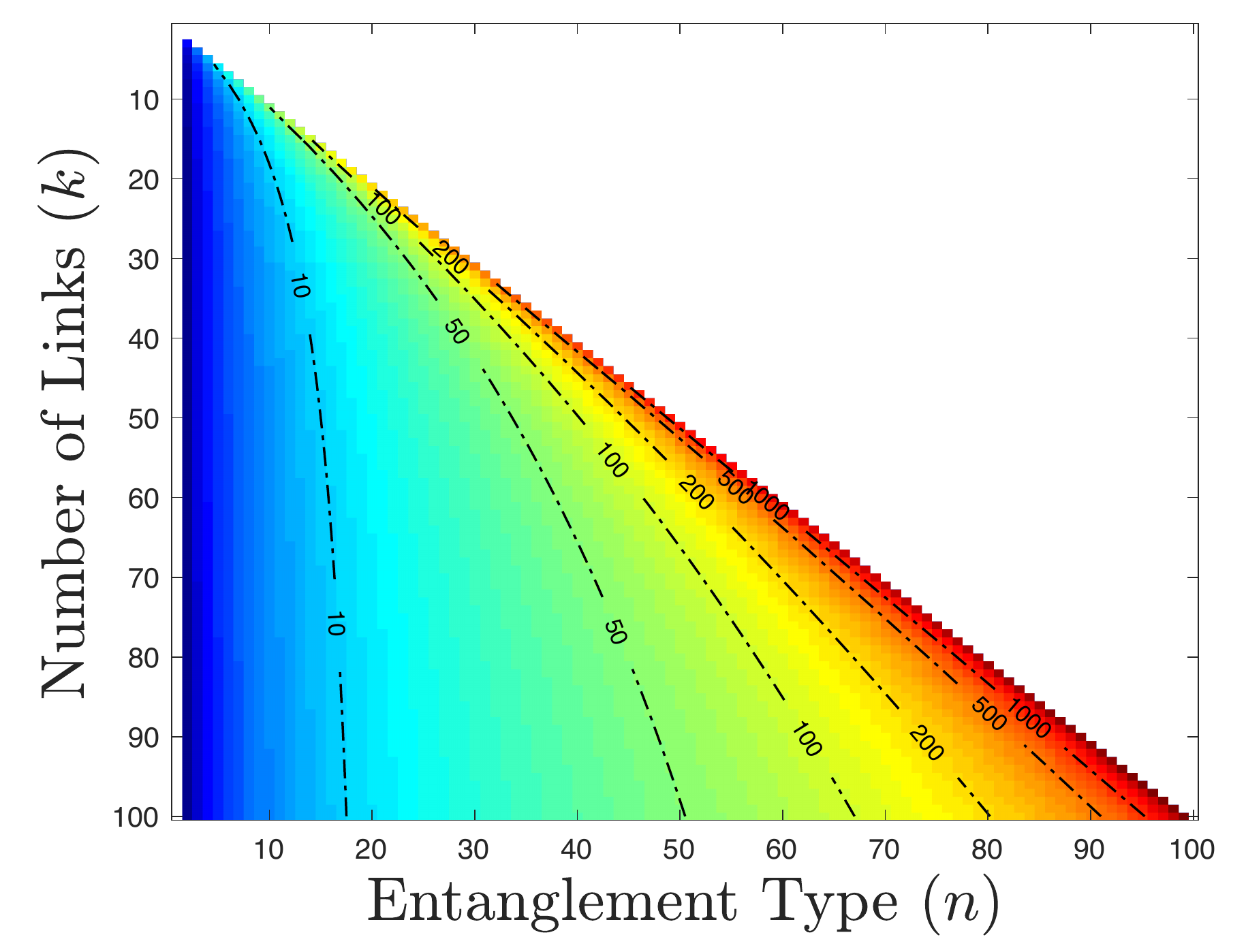}
\caption{A heatmap of the $\log$ expected number of stored qubits at the switch, $\log(\E[|\vect{Q}|])$, as a function of the number of links $k$ and entanglement type -- $n$ for $n$-partite. Overlaid contour lines are labeled with $\E[|\vect{Q}|]$ values.}
\label{fig:EQplot}
\end{figure}

Following are a proposition and proof sketch for $\E[|\vect{Q}|]$. 
The details of the proof can be found in Appendix \ref{app:EQ-long-proof}.
This proof relies on the finiteness of the first moment of the MC $X$, which comes as a consequence of $X$'s stability -- see Section \ref{subsec:finEQ}, Proposition \ref{prop:expectationX}.
\begin{proposition}[Expected Number of Stored Qubits]\hfill
\label{prop:EQlonger}

If $X$ is stable (\ie, $k>n$ -- see Proposition \ref{prop:stability}), then
\[
\E[|\vect{Q}|]=\frac{k(n-1)}{2(k-n)}.
\]
\end{proposition}
\begin{proof}[Proof Sketch]
Let $\pi$ be the stationary distribution of $X$. Assume that the MC is in steady-state at time $t=1$ (which implies that it is in steady-state at any time $t>1$).
For every mapping
$V:\{0,1,\ldots,\}^{n-1}\to [0,\infty)$ such that $\E[V(\vect{X}_1)]<\infty$, we know that
\begin{align*}
\E[V(\vect{X}_{t+1})-V(\vect{X}_t)] = 0,
\end{align*}
so that Eq. (\ref{eq:40}) holds. Define $|\vect{x} |^{(2)}=\sum_{j=1}^{n-1} x^2_j$. 
Let $A\coloneqq \sum_{\vect{x}\in R_0}  \pi(\vect{x})|\vect{x}|$ and $B\coloneqq\sum_{j=1}^{n-2} \sum_{\vect{x}\in R_j} \pi(\vect{x})|\vect{x}|$.
Note that $\E[|\vect{Q}|]\coloneqq A+B$.
To obtain the final result, we use two Lyapunov functions: first, take
$V(\vect{x})=\min\{|\vect{x}|^2,T^2\}$, and then take $V(\vect{x})=\min\{|\vect{x}|^{(2)},T\}$. After substitution into Eq. (\ref{eq:40}) and simplification, we obtain the following set of linear equations:
\begin{eqnarray*}
2(k-n)(n-1)A-2kB&=&k(n-1)+g(T)\\
2(k-n)A- 2B&=&\frac{2k(n-1)}{n}+h(T),
\end{eqnarray*}
where $g(T)$ and $h(T)$ both go to $0$ as $T\to \infty$.
When $k>n$ there is a unique solution given by 
\begin{eqnarray*}
A&=&\frac{k(n-1)(2k-n)}{2n(k-n)(k-(n-1))}\\
B&=& \frac{k(n-1)(n-2)}{2n(k-(n-1))}. 
\end{eqnarray*}
We find $\E[|\vect{Q}|]=A+B=\frac{k(n-1)}{2(k-n)}$.
\end{proof}
The proof of Proposition \ref{prop:EQlonger} (see Appendix \ref{app:EQ-long-proof}) can be considerably  shortened if instead of working with the Lyapunov functions $\min\{|\vect{x}|^2,T^2\}$ and 
$\min\{|\vect{x}|^{(2)},T\}$, one selects the functions $|\vect{x}|^2$ and $|\vect{x}|^{(2)}$, respectively. However, in addition to the stability assumption, this approach requires to assume that $\E[|\vect{Q}|^2]<\infty$ so that Lemma \ref{lemma:eq40} can be invoked; see \cite{vardoyan2019stochasticarXiv} for a proof. We conjecture that $\E[|\vect{Q}|^2]<\infty$ if the system is stable but have not been able to prove it.

Figure \ref{fig:EQplot} presents a heatmap of $\log(\E[|\vect{Q}|])$ as a function of $k$ and $n$, and an overlaid contour plot with lines labeled using $\E[|\vect{Q}|]$ values. The number of links is varied from three to 100, and for each value of $k$, $n$ is varied from two to $k-1$. Note that as $n$ increases, so does the number of stored qubits. Along the diagonal, the data points correspond to cases when $n=k-1$. In such, cases $\E[|\vect{Q}|]$ becomes large, and continues to increase as one descends down the diagonal. This plot suggests that for small $n$ (\eg, $n\leq 20$), few quantum memories are required on average for our protocol.

\section{Stability Analysis}
\label{subsec:stability}
In this section, we prove that the system is stable if and only if $k>n$. 
In the proofs that follow, we make use of some general results and several useful formulas derived in Appendices \ref{app:generic} and \ref{app:useful-formulas}.
\begin{proposition}[Stability]
\label{prop:stability}
The Markov chain $X$ is stable (ergodic) when $k>n$.
\end{proposition}
\begin{proof}
As discussed in Section \ref{sec:relwork}, we were not able to apply Foster's criterion directly to the chain $X=\{\vect{X}_t\}_{t\geq 1}$. However, we will introduce a Markov chain $Y=\{\vect{Y}_t\}_{t\geq 1}$ embedded in $X$ and prove that it is positive recurrent. Using this result, we will argue that $X$ is also positive recurrent. First, let us define some useful sets:
\begin{itemize}
\item[-] $S\coloneqq\{\vect{x}=(x_1~\ldots~x_{n-1})\in \N^{n-1}: x_1\geq 1,\ldots, x_{n-1}\geq 1\}$;
\item[-] $S^c$ -- the complementary set of $S$ in $\N^{n-1}$;
\item[-] $S^{\star}$ -- the subset of $S$ containing vectors $\vect{x}\in S$ for which $\vect{x}-\vect{1}\not \in S$. In words, if  $\vect{x}\in S^{\star}$, then at least one $x_i$ is equal to 1. For instance, if $n=3$ then $S^{\star}=\{(i,1), (1,i),i\geq 1\}$;
\item[-] $S^{\star}_j$ -- the subset of $S^{\star}$ whose vectors have exactly $j$ components equal to $1$;
\item[-] $S_j$ -- the subset of $S^c$  whose vectors have exactly $j$ components equal to $0$.
\end{itemize}
The  Markov chain $Y$ is embedded in $X$ at times when $X$ lies in the set $S$.
Following is a formal definition of $Y$:
\begin{equation}
\label{def:Y}
\vect{Y}_t=\vect{X}_{\min\{m\geq 1: \sum_{i=1}^m \ind_{\{\vect{X}_i\in S\}}=t\}}, \quad t=1,2,\ldots.
\end{equation}
Denote by  $q(\vect{x},\vect{y})$, $\vect{x}, \vect{y} \in S$,  the one-step transition probabilities of $Y$.
The non-zero one-step probability transitions of $Y$ are given by
\begin{align}
& q(\vect{x},\vect{x}+{\bf e}_i)=\frac{1}{k} \quad\hbox{ for } \vect{x}\in S, \, i=1,\ldots,n-1, \label{prob-trans1}\\
& q(\vect{x},\vect{x}-{\bf 1})= \frac{k-n+1}{k}\quad \hbox{ for } \vect{x}\in S-S^{\star}, \label{prob-trans2}\\
& q(\vect{x},\vect y) =\frac{k-n+1}{k}\,q^{\star}(\vect{x}-{\bf 1},\vect y)   \quad \hbox{ for } 
\vect{x}, \vect y\in S^{\star},\label{prob-trans3}
\end{align}
where $q^{\star}(\vect{x},\vect{y})$, $\vect{x}\in S^{c}$, $ \vect{y}\in S^{\star}$,  is the probability that  when in state $\vect{x}\in S^c$, $X$ will re-enter $S^{\star}$ through state $\vect{y}$. Clearly, ${\sum_{\vect{y}\in S^{\star}}q^{\star}(\vect{x}, \vect{y})=1}$
for all $\vect{x}\in S^c$ so that  (\ref{prob-trans1})-(\ref{prob-trans3}) define the one-step transition probabilities of $Y$ in $S$. Following, we derive an expression for $q^{\star}$.

For $j=1,\ldots,n-1$, define $E_j$ as the set of vectors $\vect{r}=(r_1~\dots~r_{n-1})$ for which at least one entry among entries $1,\dots,j$ is equal to $1$, and entries $j+1,\dots,n-1$ can take any values in $\N$. \emph{I.e.},
 \begin{align}
\label{def:Ej}
 E_j\coloneqq\left\{\vect{r}: r_i\geq 1, 1\leq i\leq j; r_l\geq 0, j+1\leq l\leq n-1, \prod_{i=1}^j (r_i-1)=0\right\}.
\end{align}
{\em Warning: Lemma \ref{lem:sum-f} below is not referenced in the paper and can be skipped. We left it in this report so that the numbering of equations throughout the paper and
the numbering of Lemmas and Propositions in Section \ref{subsec:stability} are the same as in the journal paper.}

\begin{lemma}
\label{lem:sum-f}
For $j=1,\ldots,n-1$, define
\[
f_j(\vect r)= \sum_{l=1}^j \ind_{\{r_l=1\}} \frac{(|\vect r|-1)!}{r_1! \cdots r_{n-1}!}\left( \frac{1}{k}\right)^{|\vect r|} g(r_{n-1}), \quad
\vect{r}=(r_1~\ldots~r_{n-1}).
\]
For $j=1,\ldots,n-2$ and $N_i \coloneqq r_i+\cdots+r_{n-1}$ (assume $N_{n-1}=r_{n-1}$),
\begin{equation}
\label{sum-fj}
\sum_{\vect r\in E_j} f_j(\vect r)=\sum_{l=1}^j {j\choose l} \frac{l}{k^l} \sum_{r_{l+1}\geq 2, \ldots, r_j\geq 2\atop r_{j+1}\geq 0,\cdots, r_{n-1}\geq 0}
\frac{(l-1+N_{l+1})!}{r_{l+1}! \cdots r_{n-1}!}\left( \frac{1}{k}\right)^{N_{l+1}} g(r_{n-1}),
\end{equation}
for any mapping $g:\N\to (-\infty,\infty)$ such that the r.h.s. of (\ref{sum-fj})  is finite.

In addition, (\ref{sum-fj}) holds for $j=n-1$ if $g(i)=1$ for all $i\in \N$ provided that  the r.h.s. of (\ref{sum-fj})  is finite.
\end{lemma}
See Appendix \ref{app:proof-lemSumf} for a proof.
\begin{lemma}[Derivation of $q^{\star}$]\hfill\\
\label{lem:hat-q}
For $j=1,\ldots,n-1$, let $\vect{x}_j=(0~\ldots~0~x_{j+1}~\ldots~{x_{n-1}})\in S_j$ (\ie, $x_i\geq 1$ for $i=j+1,\ldots, n-1$). 
Then,
\begin{equation}
\label{hat-q}
q^{\star}(\vect{x}_j, \vect{r}(\vect{x}_j))= {k-n+j+1\choose k-n+1}  \frac{(|\vect{r}|-1)!}{r_1! \cdots r_{n-1}!}
 \left(\frac{1}{k}\right)^{\sum_{l=1}^{n-1} r_l} \sum_{l=1}^j \ind_{\{r_l=1\}},
\end{equation}
for all $\vect r(\vect x_j)\coloneqq \vect x_j+\vect r $ with  $\vect r=(r_1~\ldots~r_{n-1})\in E_j$ (notice that $\vect r(\vect x_j)\in S^{\star}$ when
 $\vect r\in E_j$), and
 \begin{equation}
 \sum_{\vect r\in E_j} q^{\star}(\vect x_j, \vect r(\vect x_j))=1.
 \label{lem:trans-proba}
 \end{equation}
Moreover, $q^{\star}(\vect x_j, \vect x^\prime)=0$  for all $\vect x^\prime \in S^{\star}$ if $\vect x^\prime$ is not of the form $\vect r(\vect x_j)
=\vect x_j+\vect r $ with $\vect r\in E_j$.
\end{lemma}
\begin{proof}
For the following, fix $j$ in $\{1,\ldots,n-1\}$. 
It is well-known\footnote{Formula (\ref{nb-paths}) is obtained by calculating the number of ways of depositing  $i_1+\cdots+i_{n-1}$ objects in $n-1$ bins with $i_1$ objects in the first bin, $i_2$ objects in the second bin, and so on.}  that the number of possible paths to go from state $\vect{0}$ to state $(i_1~\ldots~i_l)$ with $i_1+\cdots+i_l\geq 1$ is 
\begin{equation}
\label{nb-paths}
\frac{(i_1+\cdots+i_l)!}{i_1!\cdots i_l!},
\end{equation}
where by convention $0!=1$.

Recall that $\vect{r} \in E_j$. State $\vect x_j+\vect r$ is necessarily reached from one of the states $\vect x_j+\vect r -\vect e_l$ such that $r_l=1$ with $l\in \{1,\ldots,j\}$.
Since the number of paths connecting $\vect x_j$ to $\vect x_j+\vect r-\vect e_l$, denoted as $N_l(\vect r)$,  is equal to the number of paths connecting
$(0~\ldots~0)$ to $\vect r-\vect e_l$, we may use (\ref{nb-paths}) to obtain
\[
N_l(\vect r)=\frac{(r_1+\cdots+r_{n-1}-1)!}{r_1! \cdots r_{n-1}!} = \frac{(|\vect{r}|-1)!}{r_1! \cdots r_{n-1}!}.
\]
Therefore, the number of paths connecting $\vect x_j$ to $\vect x_j+\vect r$, denoted by $N(\vect r)$, is given by
\begin{equation}
N(\vect r)=\frac{(|\vect{r}|-1)!}{r_1! \cdots r_{n-1}!}\sum_{l=1}^j \ind_{\{r_l=1\}}.
\label{Nr}
\end{equation}
On the other hand,  each of the $N(\vect r)$ paths connecting $\vect x_j$ to $ \vect r(\vect x_j)$ has the same probability of occurrence, given by
\begin{align}
\prod_{i=1}^j \frac{k-n+i+1}{k i} \left(\frac{1}{k}\right)^{\sum_{i=1}^j  (r_i-1)} 
\left(\frac{1}{k}\right)^{\sum_{i=j+1}^{n-1} r_i} ={k-n+j+1\choose k-n+1} \left(\frac{1}{k}\right)^{|\vect{r}|}.
\label{prob-path}
\end{align}
Therefore, cf. (\ref{Nr})-(\ref{prob-path}),
\[
q^{\star}(\vect x_j, \vect r(\vect x_j))= {k-n+j+1\choose k-n+1} \frac{(|\vect{r}|-1)!}{r_1! \cdots r_{n-1}!}
 \left(\frac{1}{k}\right)^{\sum_{l=1}^{n-1} r_l}  \sum_{l=1}^j \ind_{\{r_l=1\}},
\]
which establishes (\ref{hat-q}).

Next, define $D_j=\sum_{\vect r\in E_j} q^{\star}(\vect x_j, \vect r(\vect x_j))$.  Using (\ref{hat-q}), we obtain
\begin{align}
D_j&={k-n+j+1\choose k-n+1}\sum_{\vect r\in E_j} \frac{(|\vect{r}|-1)!}{r_1! \cdots r_{n-1}!}\left(\frac{1}{k}\right)^{|\vect{r}|}\sum_{l=1}^j \ind_{\{r_l=1\}}.
\label{eq:dj}
\end{align}
We can rewrite the sum above as follows:
\begin{align*}
\sum_{l=1}^j {j \choose l}\frac{l}{k^l}
\sum_{r_{l+1}\geq 2,\dots, r_j\geq 2\atop r_{j+1}\geq 0,\dots,r_{n-1}\geq 0} \frac{(r_{l+1}+\cdots+r_{n-1}+l-1)!}{r_{l+1}! \cdots r_{n-1}!}\left(\frac{1}{k}\right)^{\sum\limits_{i=l+1}^{n-1}r_i},
\end{align*}
which, using the definition of $F_m(g;J,L)$ in (\ref{def-Fm}), equals
\begin{align*}
\sum_{l=1}^j {j \choose l}\frac{l}{k^l}F_{j-l}(1;n-1-l,l-1).
\end{align*}
Using Lemma \ref{lem:Gm}, this equals
\begin{align*}
\frac{1}{k}\sum_{l=1}^j {j \choose l}(-1)^{l+1}lF_{0}(1;n-1-l,0).
\end{align*}
Substituting this into (\ref{eq:dj}), we obtain
\begin{align}
D_j&= {k-n+j+1\choose k-n+1}\frac{1}{k}\sum_{l=1}^j {j \choose l}(-1)^{l+1}lF_{0}(1;n-1-l,0)
={k-n+j+1\choose k-n+1}\sum_{l=1}^j{j\choose l}\frac{(-1)^{l+1}l}{k-n+l+1},
\end{align}
where the latter equality is obtained by using (\ref{F0:g1}). Note that $D_1=1$. Since
\[
D_{j+1}-D_j=\frac{(k-n+j+1)!}{(k-n+1)! (j+1)!}\sum_{l=1}^{j+1} {j+1\choose l}(-1)^{l+1} l =0,
\]
where the last equality is true from (\ref{poly-result}), this proves that $1=D_1=D_2=\cdots=D_{n-1}$. This proves (\ref{lem:trans-proba}).

The last statement in the lemma is clearly true by definition of $\vect x_j$,  $S^{\star}$ and $E_j$.
\end{proof}
\begin{remark}
\label{rem:hat-q}
Note that $q^\star(\vect x_j, \vect r(\vect x_j))$ given in the r.h.s. of (\ref{hat-q}) does not depend on $x_{j+1},\ldots,x_{n-1}$.
Moreover, since the channels are statistically indistinguishable, the r.h.s. of (\ref{hat-q}) gives the transition probability from any $\vect z_j\in S_j$ to  $S^\star$.

Another interpretation of  (\ref{hat-q})  is to say that with probability 
\[
 {k-n+j+1\choose k-n+1}  \frac{(r_1+\cdots+ r_{n-1}-1)!}{r_1! \cdots r_{n-1}!}\left(\frac{1}{k}\right)^{\sum_{l=1}^{n-1} r_l} \sum_{l=1}^j \ind_{\{r_l=1\}},
\]
the MC $X$ will stay in $S^c$ for $|\vect r|$ time-steps, $\vect r\in E_j$, given that it is currently located in the set $S_j$ (note that $|\vect r|\geq j$ when
$\vect r\in E_j$).
\end{remark}
Next, we choose a suitable Lyapunov function and examine the drift of $Y$.
For any  (Lyapunov) function $V:\N^{n-1}\to [0,\infty)$, let  $\Delta_V(\vect x):=\E[V(\vect{Y}_{n+1}) - V(\vect x)\,|\, \vect{Y}_n=\vect x]$, $\vect x=(x_1~\ldots~x_{n-1})\in S$,  be  the drift of $Y$ associated with $V$. We will consider the function 
\begin{equation}
\label{def:V}
V(\vect x)=\sum_{i=1}^{n-1} x_i^2+ b\sum_{1\leq i<l\leq n-1} x_i x_l. 
\end{equation}
By writing $V(\vect x)$ as
\begin{equation}
\label{def:V-2}
V(\vect x)=\frac{1}{n-2}\sum_{1\leq i<l\leq n-1} (x_i-x_l)^2 +\left(\frac{2}{n-2}+b\right) \sum_{1\leq i<l\leq n-1} x_i x_l,
\end{equation}
we see that it is nonnegative when $b\geq -\frac{2}{n-2}$ for all $\vect x\in S$.
\begin{proposition}[Drift of $Y$]\hfill
\label{prop:drift}

For $\vect x=(x_1~\ldots~x_{n-1})\in S-S^{\star}$,
\begin{align}
\label{prop:delta-1}
k\times \Delta_V(\vect x)=&- (k-n) (2+b(n-2))\sum_{i=1}^{n-1} x_i
+(n-1)\left(k-n+2+ b\frac{(k-n+1)(n-2)}{2}\right).
\end{align}
For $\vect x_j'=(1~\ldots~1~x_{j+1}~\ldots~x_{n-1})\in S^{\star}_j$  (\ie, $x_i\geq 2$ for ${i=j+1},\ldots, {n-1}$),
\begin{align}
k\times \Delta_V(\vect{x_j'}) =&
\Biggl[(k-n+1)(2+b(k-1))\sum_{i=2}^{j+1} \frac{1}{k-n+i}
-(k-n)(2+b(n-2))\Biggr]\sum_{i=j+1}^{n-1}x_i + \delta_j,
\label{prop:delta-2}
\end{align}
where $\delta_j$ is independent of $x_{j+1},\ldots,x_{n-1}$ and such that 
${\max_{1\leq j\leq n-2}|\delta_j|<\infty}$.
\end{proposition}
The proof of this proposition is in Appendix \ref{app:proof-drift}.

Using these results, we can use Foster's criterion to derive a sufficient stability condition for $X$ and complete the proof of Proposition \ref{prop:stability}.
For the following, $n$ and $k$ are fixed with $3\leq n\leq k$. Let $\epsilon>0$.
First take $\vect x=(x_1~\ldots~ x_{n-1})\in S-S^{\star}$. From (\ref{prop:delta-1}) we see that
 for $\sum_{i=1}^{n-1}x_i$ large enough, $\Delta_V(\vect x)$ can be made strictly less than $-\epsilon$ when $k>n$
and $b>-2/(n-2)$. This shows that, under these conditions, Foster's criterion applies to states in $S-S^{\star}$.

Consider now states in $S^{\star}$. Let $\vect x_j\in S^{\star}$ be a state that has exactly $j$  entries equal to $1$ with $j=1,\ldots,n-2$ (the case $j=n-1$, \ie, state $\vect{x}_{n-1}=(1~\dots~1)$, will be addressed below).
Since any permutation of the entries of  $\vect x\in S$  is equally likely to occur, without loss of generality we can consider the state 
$\vect x_j'=(1~ \ldots~ 1~ x_{j+1}~\ldots~x_{n-1})$ with  $x_i\geq 2$ for $i=j+1,\ldots,n-2$.

Assume that $b=-\frac{2(1-\alpha)}{n-2}$. Note that $b>-\frac{2}{n-2}$ if $\alpha>0$.
With this choice of $b$ and upon setting $k-n=m$, the coefficient of $\sum_{i=j+1}^{n-1}x_i$ in (\ref{prop:delta-2}), denoted by $C_j$,    $j=1,\ldots,n-2$, becomes
\begin{align}
C_j&=2\left(m+1-\left(\frac{1-\alpha}{n-2}\right)(m^2+nm+n-1)\right)\sum_{i=2}^{j+1} \frac{1}{m+i}
-2m\left(1-\frac{1-\alpha}{n-2}\times (n-2)\right)
\nonumber\\
&=-\frac{2}{n-2}\left((m+1)^2-\alpha(m^2 +mn+n-1)\right)\sum_{i=2}^{j+1} \frac{1}{m+i}-2m\alpha.
\label{Cj-2}
\end{align}
Assume now that $0<\alpha<\frac{1}{n-1}$. Then,
\begin{equation}
\label{proof:inq}
(m+1)^2-\alpha(m^2 +mn+n-1)>\frac{m(m+1)(n-2)}{n-1}>0,
\end{equation}
for all $m>0$. Hence, for $j=1,\ldots,n-2$, $m\geq 1$ (or equivalently for $k>n$)
\[
C_j< -\frac{2m(m+1)}{n-1}\sum_{i=2}^{j+1}\frac{1}{m+i}-2m\alpha < 0.
\]
This shows from (\ref{prop:delta-2}) and the finiteness of $\delta_j$ that for $\sum_{i=j+1}^{n-1}x_i$ large enough, $\Delta_V(\vect x_j)<-\epsilon$.

Define the finite set $F_M=\{\vect x\in S: \sum_{i=1}^{n-1} x_i\leq M\}$. Note that $\vect x_{n-1}=(1,\ldots,1)\in F_M$ as long as $M\geq n-1$, which we will assume to be true.
The above shows that, if $k>n$, there exists $M>0$ such that $\Delta_V(\vect x)<-\epsilon$ for all
$\vect x\in F^c_M=\{\vect x\in S: \sum_{i=1}^{n-1} x_i> M\}$, where $b=-\frac{2\alpha}{n-2}$ with $0<\alpha<\frac{1}{n-1}$ in the definition of $V$.
On the other hand, $\E[V(\vect{Y}_{n+1})|\vect{Y}_n=\vect x]=\sum_{\vect y\in S} q(\vect x,\vect y)V(\vect{y}) <\infty$ for all  $\vect x\in F_M$, as   from any state $\vect x$ only a finite number of states 
$\vect y$ are reachable in one time unit (states $\vect x+\vect e_i$, $i=1,\ldots,n-1$ when $\vect x\in S$ and state $\vect x-{\bf 1}$ when $\vect x\in S-S^{\star}$),
which implies the finiteness of $\E[V(\vect{Y}_{n+1})|\vect{Y}_n=\vect x]$ when $\vect x$ belongs to the finite set $F_M$. Hence,
Foster's criterion \cite[pp. 167-168]{Bremaud99} applies to the  irreducible Markov chain $Y$ (the irreducibility of $Y$ is inherited from the irreducibility of $X$;  irreducibility can also be checked directly from 
(\ref{prob-trans1})-(\ref{prob-trans2})), which shows that $Y$ is positive recurrent on $S$ when $k>n$ and so is $X$ since $Y$ is embedded in $X$. Since $X$ is also aperiodic, this proves that it is stable when $k>n$.
\end{proof}

Above, we proved that the system is stable when $k>n$. Next, we will show that this condition is not only sufficient, but also necessary.
\begin{proposition}[Instability]
The Markov chain $X$ is unstable when $k=n$.
\label{prop:instability}
\end{proposition}
\begin{proof}
Let $\vect x=(x_1~\ldots~ x_{n-1})\in R$ be the state of the DTMC $X$. Denote the transition probabilities of $X$ using $p_{\vect{x}, \vect{y}}$.

It is shown in \cite[Theorem 2]{Sennott1985} that if  there exists a positive integer $M$ such that $\sum_{|\vect y|<|\vect x|}p_{\vect x,\vect y}(y_i-x_i)\geq -M$ for
$1\leq i\leq n-1$, $\vect x\in R$, and if $\sum_{i=1}^{n-1} \sum_{\vect y} p_{\vect x,\vect y}(y_i-x_i) \geq 0$ for all $\vect x\in R$, then the DTMC is not ergodic.

Let us apply this result to $X$. Take first $\vect x=(x_1~\ldots~x_{n-1})$ with $x_i\geq 1$ for all $i$. The only state such that $|\vect x|>|\vect y|$ and into which $X$
can jump in one time step is $\vect x-{\bf 1}$. Therefore,
\begin{align*}
\sum_{|\vect y|<|\vect x|}p_{\vect x,\vect y}(y_i-x_i)=\frac{k-n+1}{k}\times (x_i-1-x_i)=-\frac{k-n+1}{k},
\quad \text{for all } i=1,\ldots, n-1.
\end{align*}
If $\vect x$ has at least one zero entry, $X$ cannot go  in one time step from $\vect x$ to a state $\vect y$ such that $|\vect x|>|\vect y|$ and for these states
any $M\geq 0$ works.  Therefore, Theorem 2 in \cite{Sennott1985} holds with (for instance) $M=\frac{k-n+1}{k}$. Note that $M>0$ since $k\geq n$.

Let us verify the second condition in Theorem 2. Take $\vect x=(x_1~\ldots~x_{n-1})$ with $x_i\geq 1$ for all $i$. Then,
\begin{align*}
&\sum_{i=1}^{n-1}\sum_{\vect y\in \N^{n-1}} p_{\vect x,\vect y}(y_i-x_i)=
\sum_{i=1}^{n-1}\left(\frac{k-n+1}{k} (x_i-1-x_i) + \frac{1}{k}
(x_i+1-x_i)\right)=-\frac{(n-1)(k-n)}{k},
\end{align*}
which is nonnegative if and only if  $k=n$ (under the constraint that $k\geq n$).

Take now  $\vect x_j=(0~\ldots~0~x_{j+1}~\ldots~x_{n-1})$  with $x_i\geq 1$ for $i=j+1,\ldots,n-1$, $j=1,\ldots,n-1$. We have
\begin{align*}
\sum_{i=1}^{n-1} \sum_{\vect y\in \N^{n-1}} p_{\vect x,\vect y}(y_i-x_i)&=
\sum_{i=1}^{n-1}\left( \sum_{l=1}^j p_{\vect x, \vect x+\vect e_l}(x_i+\ind_{\{l=i\}}-x_i)+\sum_{l=j+1}^{n-1}p_{\vect x, \vect x+\vect e_l}(x_i+\ind_{\{l=i\}}-x_i)\right)\\
&= \sum_{i=1}^{n-1}\left( \sum_{l=1}^j \frac{k-n+j+1}{kj} \ind_{\{l=i\}}+\sum_{l=j+1}^{n-1}\frac{1}{k}\ind_{\{l=i\}}\right),
\end{align*}
which is nonnegative for all $k\geq n$. We conclude that $X$ is unstable when $k=n$.
\end{proof}
\section{Finiteness of the First Moment}
\label{subsec:finEQ}
In this section, we show that when the system is stable, \ie, $k>n$, then $\E[|\vect{Q}|]<\infty$.  The finiteness of $\E[|\vect{Q}|]$ is used in the proof of Propositions 
\ref{capacity-n}  and \ref{prop:EQlonger}.

Recall from Section \ref{subsec:stability} that when $k>n$, both Markov chains $X$ and $Y$ are ergodic, with  stationary elements $\vect{Q}$ and $\vect{Y}$, respectively,
that is $\vect{X}_t\to \vect{Q}$ a.s. and $\vect{Y}_t\to \vect{Y}$ a.s.  as $t\to\infty$.

We will now argue that when $k>n$, $\E[|\vect{Y}|]< \infty$. From this, it will follow that $\E[|\vect{Q}|]<\infty$.

In case of potential confusion with the notation $|\vect{v}|=\sum_{i=1}^n v_i$ if $\vect{v}=(v_1,\ldots,v_n)$,
we will denote by $\abs(x)$ the absolute value of any real number $x$.

\begin{proposition}
\label{prop:expectationY}
$\E[|\vect{Y}|]<\infty$ when $k>n$.
\end{proposition}
Proof of this proposition is in Appendix \ref{app:expYproof}.
\begin{proposition}
\label{prop:expectationX}

$\E[|\vect{Q}|]<\infty$ when $k>n$.
\end{proposition}
\begin{proof}
Let $\{\tau_i\}_{i\geq 1}$ be the successive times when $X$ enters $S^c$ and let  $\{T_i\}_{i\geq 1}$ be the successive times when $X$ leaves $S^c$ (\ie, enters $S$).
Since the MCs $X =\{\vect{X}_t\}_{t\geq 1}$ and $Y=\{\vect{Y}_t\}_{t\geq 1}$ are both ergodic when $k>n$, and since we will let $t\to\infty$ in the following, we may assume 
without loss of generality that  $\vect{X}_1\in S$ (implying that $\vect{Y}_1=\vect{X}_1$ -- see (\ref{def:Y})), so that $1<\tau_1<T_1<\tau_2<T_2<\cdots$.
Recall that $S^c=\cup_{j=1}^{n-1}S_j$.
With these definitions, $\vect{X}_t\in S^c$ when $t\in \cup_{i\geq 1}[\tau_i,T_i)$ and $\vect{X}_t\in S$ when  $t\in \cup_{i\geq 0}[T_i, \tau_{i+1})$ with $T_0=1$. 

Define
\begin{equation}
\label{def:Mt}
M_t= \sum_{i=1}^t {\ind}_{\{\vect{X}_i\in S\}}.
\end{equation}
As an illustration, assume that  $X_1\in S$, $X_2\in S^c$,  $X_3\in S^c$, $X_4\in S^c$, and $X_5\in S$.  Then,
$M_1=1$, $M_2=1$, $M_3=1$, $M_4=1$, and $M_5=2$.
Let us focus on the difference $|\vect{X}_t|- |\vect{Y}_{M_t}|$. By construction of the MC $Y$,  $\vect{X}_t=_{st}\vect{Y}_{M_t}$ when $\vect{X_t}\in S$, that is,
when $t\in \cup_{i\geq 0}[T_i, \tau_{i+1})$.  
%
Therefore,
\begin{align}
\E[\hbox{abs}( |{\vect{X}}_t|-|{\vect{Y}}_{M_t}|) ]&=\sum_{i\geq 1} \E\left[ \hbox{abs}( |{\vect{X}}_t|-|{\vect{Y}}_{M_t}|) {\ind}_{\{\tau_i\leq t<T_i\}}\right]\nonumber\\
&= \sum_{i\geq 1} \E[ \hbox{abs}(|{\vect{X}}_t|-|{\vect{Y}}_{M_t}| ) \,|\, \tau_i\leq t<T_i] \Prob  (\tau_i\leq t<T_i).
\label{difference1}
\end{align}
Conditioned on $\vect{X}_{\tau_i}=\vect x_j\in S_j$,  $j=1,\ldots,n-1$, \[
|\vect{X}_t|=|\vect x_j|+ t-\tau_i,\quad |\vect{Y}_{M_t}|= |\vect x_j|+ n-1,
\]
for $t=\tau_i,\tau_i+1, \ldots,T_i-1$,  since $|\vect{X}_t|$ increases by $1$ at times $t=\tau_i+1, \ldots, T_i-1$, $Y_{M_{\tau_i}}= |\vect x_j|+ n-1$, and 
$M_{\tau_i}=\ldots= M_{T_i-1}$ by definition of $M_t$.
Hence, by (\ref{difference1}),
\begin{align}
\hspace{-1.7mm}
\E[\hbox{abs}( |{\vect{X}}_t|&-|{\vect{Y}}_{M_t}|) ] 
=\sum_{i\geq 1}\sum_{j=1}^{n-1}
\sum_{\vect x_j\in S_j}\E[| t-\tau_i -(n-1)| \,|\,\vect{X}_{\tau_i}=\vect x_j, \tau_i\leq t<T_i]
\Prob  (\vect{X}_{\tau_i}=\vect x_j, \tau_i\leq t<T_i)\nonumber\\
&\leq\sum_{i\geq 1}\sum_{j=1}^{n-1}\sum_{\vect x_j\in S_j}\left(\E[T_i-\tau_i\,|\,\vect{X}_{\tau_i}=\vect x_j, \tau_i\leq t<T_i] +n-1\right)
\Prob  (\vect{X}_{\tau_i}=\vect x_j, \tau_i\leq t<T_i).
\label{difference2}
\end{align}
Let us focus on $\Psi_j:=\E[T_i-\tau_i\,|\,\vect{X}_{\tau_i}=\vect x_j, \tau_i\leq t<T_i]$ in (\ref{difference2}).
%
Assume for the time being that $X$ is a continuous-time Markov chain, with $\mu$ the rate at which a link generates an entanglement.
The time spent by $X$ in $S_l$  is an exponential rv  $Z_l$ with rate $\mu(k-n+l+1)$, so that the  time spent in 
$S^c$ starting from $S_j$ is $\sum_{l=1}^j  Z_l$, and the expected time is given by $\sum_{l=1}^j \frac{1}{\mu(k-n+l+1)}$.
We come back to the discrete-time MC by uniformizing the continuous-time Markov chain at rate $k\mu$, which yields
\begin{equation}
\Psi_j=\sum_{l=1}^j \frac{k}{k-n+l+1}.
\label{sign3}
\end{equation}
Define 
\begin{align*}
D_0&=\max_{1\leq j\leq n-1}\left\{\sum_{l=1}^{j} \frac{k}{k-n+l+1}+n-1\right\}
=\sum_{l=1}^{n-1}\frac{k}{k-n+l+1}+n-1<\infty.
\end{align*}
Combining (\ref{difference2}) and (\ref{sign3}), we obtain 
\begin{align}
\E[\abs(|\vect{X}_t|-|{\vect{Y}}_{M_t}|)]&\leq D_0
\sum_{i\geq 1}\sum_{j=1}^{n-1}\sum_{\vect x_j\in S_j}\Prob  (\vect{X}_{\tau_i}=\vect x_j, \tau_i\leq t<T_i)
= D_0\sum_{i\geq 1}\sum_{j=1}^{n-1} \Prob  (\vect{X}_{\tau_i}\in S_j,\tau_i\leq t<T_i)\nonumber\\
&= D_0\sum_{i\geq 1}\Prob (\vect{X}_{\tau_i}\in S^c, \tau_i\leq t<T_i)\quad \hbox{as } S^c =\cup_{j=1}^{n-1} S_j,\nonumber\\
& =   D_0\sum_{i\geq 1}\Prob (\tau_i\leq t<T_i)\leq D_0,
\label{differenceX-Y-0}
\end{align}
since $ \sum_{i\geq 1}\Prob (\tau_i\leq t<T_i)=\E\left[\sum_{i\geq 1} {\ind}_{\{\tau_i\leq t<T_i\}}\right] \leq 1$.  

By Fatou's lemma,
\[
\E[\liminf_t \abs(|\vect{X}_t|-|{\vect{Y}}_{M_t}|)]\leq \liminf_t \E[ \abs(|\vect{X}_t|-|{\vect{Y}}_{M_t}|)]
\]
so that
\begin{equation}
\label{inq-abs}
\E[\abs(|\vect{Q}|-|{\vect{Y}}|)] \leq D_0,
\end{equation}
by using (\ref{differenceX-Y-0}) and the fact that $\vect{X}_t\to \vect{Q}$ a.s. and $\vect{Y}_{M_t}\to \vect{Y}$ a.s. as $t\to\infty$ when $k>n$
(Hint: $M_t\to\infty$ a.s. as $t\to\infty$, since $X$ is irreducible and recurrent on $S\cup S^c$, which implies, in particular, that it visits $S$ infinitely often).

The inequality $x\leq |x|$ which holds for any real number $x$, yields
\[
|\vect{Q}|\leq \abs(|\vect{Q}|-|{\vect{Y}}|) + |\vect{Y}| \quad \hbox{a.s.}
\]
which in turn gives, by using (\ref{inq-abs}),
\[
\E[|\vect{Q}|]\leq D_0+E[|{\vect{Y}}|] <\infty
\]
from Proposition \ref{prop:expectationY}. This concludes the proof.
\end{proof}
%

\section{Conclusion}
\label{sec:conclusion}
We analyze an assembly-like stochastic queueing system with one central node serving multiple users in a star topology. This system is analogous to a quantum switch serving $n$-partite maximally entangled states to sets of users from a total of $k\geq n$ users. When the switch has infinite memory, we find that the process is stable if and only if $k>n$. We derive closed-form expressions for the switch capacity and expected number of qubits in memory at the switch, under the assumptions that all links are identical, link-level entanglement generation is a Poisson process, and that entanglement swapping operations are instantaneous but may fail. We find that for the protocol in which any $n$ users wish to share an entangled state, memory requirements for the switch are generally low as long as $n$ is not too large, although this conclusion may differ in the presence of user demands. These results are of interest to quantum communication, while at the same time make a novel contribution to queueing theory.

For possible future directions, it may be of interest to $(i)$ analyze a variant of this system wherein the switch buffer is finite, $(ii)$ consider a variant where the links are heterogeneous -- \ie, link-level entanglement generation rates differ from link to link; and $(iii)$ incorporate finite memory coherence times into the model. The last point may have further implications in that when decoherence is introduced into the system, we must find a way to model the fact that quantum state fidelity decreases with time (recall that in our model, we assume that all states, if generated successfully, have unit fidelity). These extensions are nontrivial, but we anticipate that our work will serve as either a starting or a comparison basis for future work on this problem.
\begin{acks}
This research is supported in part by the National Science Foundation under grants ECCS-1640959 and CNS-1617437.
The authors are thankful to Alain Jean-Marie (Inria) for stimulating discussions during the course of this work.
\end{acks}

\bibliographystyle{ACM-Reference-Format}
\bibliography{qnet}

\appendix

\section{General Results}
\label{app:generic}
For a non-negative vector $\vect{n}=(n_1~\dots~n_J)\geq 0$, define
\begin{align}
F_m(g;J,L)=\sum_{n_1\geq 2, \ldots, n_m\geq 2\atop n_{m+1}\geq 0, \ldots, n_J\geq 0}\frac{(L+|\vect{n}|)!}{n_1!\cdots n_J!} \left(\frac{1}{k}\right)^{|\vect{n}|}g(n_J),
\label{def-Fm}
\end{align}
where $0\leq m\leq J$ and $L\geq 0$.  
Note that
\begin{equation}
\label{def:F0JL}
F_0(g;J,L)=\sum_{n_1\geq 0,\ldots, n_J\geq 0}\frac{(L+n_1+\cdots+n_J)!}{n_1!\cdots n_J!} \left(\frac{1}{k}\right)^{n_1+\cdots+n_J}g(n_J),
\end{equation}
and 
\begin{equation}
\label{def:FJJL}
F_J(g;J,L)=\sum_{n_1\geq 2,\ldots, n_J\geq 2}\frac{(L+n_1+\cdots+n_J)!}{n_1!\cdots n_J!} \left(\frac{1}{k}\right)^{n_1+\cdots+n_J}g(n_J).
\end{equation}
Of particular interest will be the mappings
\begin{equation}
\label{def:mappings-g}
g_1(i)\coloneqq 1,\quad g_2(i)\coloneqq i,\quad g_3(i)\coloneqq i-1. 
\end{equation}
For later use, notice that  
\begin{eqnarray}
F_0(g_1;J,L)&=& \frac{L! k^{L+1}}{(k-J)^{L+1}}, \label{F0:g1}\\
F_0(g_2;J,L)&=& \frac{(L+1)! k^{L+1}}{(k-J)^{L+2}},\label{F0:g2}\\
F_0(g_3;J,L)&=& -\frac{L! k^{L+1}(k-J-L-1)}{(k-J)^{L+2}}, \label{F0:g3}
 \end{eqnarray}
for $k>J$ and $L\geq 0$, by applying formulas (\ref{gen-formula}), (\ref{app:diff}) and (\ref{app:formula-0}) in Appendix \ref{app:useful-formulas}, respectively, with $z=\frac{1}{k}$.

\begin{lemma}
\label{lem:recursion} 
Let $g:\N\to (-\infty,+\infty)$ be such that $F_0(g;J,L)$ is finite for $k>J$ and $L\geq 0$ 
Then, for every $m=1,\ldots,J-1$, $F_m(g;J,L)$ is finite for $k>J$ and $L\geq 0$, and satisfies the recursion 
\begin{equation}
\label{recursion-Fm}
F_m(g;J,L)=F_{m-1}(g;J,L)- F_{m-1}(g;J-1,L)-\frac{1}{k} F_{m-1}(g;J-1,L+1).
\end{equation}
In addition, (\ref{recursion-Fm}) holds for $m=n-1$ if $g(i)=1$ for all $i\in \N$.
\end{lemma}
\begin{proof} 
Assume first that  $m=1,\ldots,J-1$.
Define $N_{\setminus n_m} \coloneqq n_1+\cdots+n_{m-1}+n_{m+1}+\cdots n_J$.
 From (\ref{def-Fm}), we have
\begin{align}
F_m(g;J,L)&=\sum_{n_1\geq 2, \ldots,n_{m-1}\geq 2\atop n_m\geq 0,\dots, n_J\geq 0} \frac{(L+n_1+\cdots+n_J)!}{n_1!\cdots n_J!} \left(\frac{1}{k}\right)^{n_1+\cdots+n_J}
g(n_J)\nonumber\\
&\qquad-\hspace{-3mm} \sum_{n_1\geq 2, \ldots,n_{m-1}\geq 2\atop n_{m+1}\geq 0,\dots, n_J\geq 0} \frac{(L+N_{\setminus n_m})!}{n_1!\cdots n_{m-1}! n_{m+1}!\ldots n_J!} \left(\frac{1}{k}\right)^{N_{\setminus n_m}}g(n_J)\nonumber\\
&\qquad-\frac{1}{k} \sum_{n_1\geq 2, \ldots,n_{m-1}\geq 2\atop n_{m+1}\geq 0,\dots, n_J\geq 0} \frac{(L+1+N_{\setminus n_m})!}{n_1!\cdots n_{m-1}! n_{m+1}!\ldots n_J!} \left(\frac{1}{k}\right)^{N_{\setminus n_m}}g(n_J)\nonumber\\
&= F_{m-1}(g;J,L)-F_{m-1}(g;J-1,L)-\frac{1}{k}F_{m-1}(g;J-1,L+1),
\label{proof-recursion}
\end{align}
which establishes (\ref{recursion-Fm}). The recursion (\ref{recursion-Fm}) together with the finiteness of  $F_0(g;J,L)$ for $k>J$ and $L\geq 0$ yield the finiteness
of $F_m(g;J,L)$  for $0\leq m<J<k$ and $L\geq 0$.

Via the same analysis it is easily obtained that (\ref{recursion-Fm}) holds for $m=J$ when $g(i)=1$ for all $i\in \N$.
\end{proof}
For $a\in \{-1,0,1,\ldots\}$, define
\begin{equation}
\label{app:Gm0}
G_m(g;a,n)\coloneqq \sum_{l=1}^m {m\choose l} \frac{l}{k^l} F_{m-l}(g;n-l-1,l+a), \quad m=1,\ldots,n-1.
\end{equation}
\begin{lemma}
\label{lem:Gm}
For $m=1,\ldots,n-2$,
\begin{equation}
\label{app:Gm}
G_m(g;a,n)=\frac{1}{k}\sum_{l=1}^m {m\choose l}(-1)^{l+1}l  F_0(g;n-l-1,1+a).   
\end{equation}
In addition, (\ref{app:Gm}) holds for $m=n-1$ if $g(i)=1$ for all $i\in \N$.
\end{lemma}
\begin{proof}We will use the identity
\begin{align}
\label{basic-id}
{m+1\choose l}l=\frac{m+1}{m}\left({m\choose l}l+{m\choose l-1}(l-1)\right).
\end{align}
Throughout the proof we will skip the first argument $g$ of the mappings $F$ and $G$, to simplify the notation. 
We have
\begingroup
\allowdisplaybreaks
\begin{align}
G_{m+1}&(a,n) =
 \sum_{l=1}^m {m+1\choose l} \frac{l}{k^l} F_{m+1-l}(n-l-1,l+a)
 + \frac{m+1}{k^{m+1}} F_{0}(n-m-2,m+1+a)\nonumber\\
&=\frac{m+1}{m}\Biggl[  \sum_{l=1}^m {m\choose l} \frac{l}{k^l}  F_{m+1-l}(n-l-1,l+a)
+ \sum_{l=2}^m {m\choose l-1}\frac{l-1}{k^l}F_{m+1-l}(n-l-1,l+a)\Biggr]     \nonumber\\
&\quad+ \frac{m+1}{k^{m+1}} F_{0}(n-m-2,m+1+a) \quad \hbox{by using } (\ref{basic-id})\nonumber\\
&=\frac{m+1}{m}\Biggl[   \sum_{l=1}^m {m\choose l} \frac{l}{k^l}  F_{m+1-l}(n-l-1,l+a)
+ \sum_{l=1}^{m-1} {m\choose l}\frac{l}{k^{l+1}}F_{m-l}(n-l-2,l+a+1)\Biggr]     \nonumber\\
& \quad+\frac{m+1}{k^{m+1}} F_{0}(n-m-2,m+1+a)\nonumber\\
&=\frac{m+1}{m}\Biggl[   \sum_{l=1}^m {m\choose l} \frac{l}{k^l}  F_{m+1-l}(n-l-1,l+a)
+ \sum_{l=1}^m {m\choose l}\frac{l}{k^{l+1}}F_{m-l}(n-l-2,l+a+1)\Biggr] \nonumber\\
&=\frac{m+1}{m}\Biggl[  \sum_{l=1}^m {m\choose l} \frac{l}{k^l}  F_{m-l}(n-l-1,l+a)
-  \sum_{l=1}^m {m\choose l} \frac{l}{k^l} F_{m-l}(n-l-2,l+a)\nonumber\\
&\quad-\frac{1}{k} \sum_{l=1}^m {m\choose l} \frac{l}{k^l} F_{m-l}(n-l-2,l+a+1)
+ \sum_{l=1}^m {m\choose l}\frac{l}{k^{l+1}}F_{m-l}(n-l-2,l+a+1)\Biggr]
~~ \hbox{using } (\ref{recursion-Fm}) \nonumber\\
&=\frac{m+1}{m}\Biggl[  \sum_{l=1}^m {m\choose l} \frac{l}{k^l}  F_{m-l}(n-l-1,l+a)
-  \sum_{l=1}^m {m\choose l} \frac{l}{k^l} F_{m-l}(n-l-2,l+a)\Biggr]\nonumber\\
&=\frac{m+1}{m}\left(G_m(a,n)-G_m(a,n-1)\right).
\label{app-generic:eq1}
\end{align}
\endgroup
Letting $m=1$ in (\ref{app:Gm0}), we get $G_1(a,n)=\frac{1}{k}F_0(n-2,1+a)$, which is equal to the r.h.s. of  (\ref{app:Gm}) when $m=1$. Assume that (\ref{app:Gm}) is true for $m=2,\ldots,M$; let us show that it is still true for $m=M+1$. Using the induction hypothesis in the r.h.s. of (\ref{app-generic:eq1}) yields
\begingroup
\allowdisplaybreaks
\begin{align}
G_{M+1}&(a,n)=  \frac{M+1}{kM}\Biggl[ \sum_{l=1}^M {M\choose l}(-1)^{l+1}l F_0(n-l-1,1+a)
-\sum_{l=1}^M {M\choose l}(-1)^{l+1}l F_0(n-l-2,1+a)\Biggr]\nonumber\\
&=\frac{M+1}{kM}\Biggl[ \sum_{l=1}^M {M\choose l}(-1)^{l+1}l F_0(n-l-1,1+a)
-\sum_{l=1}^{M-1} {M\choose l}(-1)^{l+1}l F_0(n-l-2,1+a)\Biggr]\nonumber\\
&\quad-\frac{M+1}{k}(-1)^{M+1}F_0(n-M-2,1+a)\label{chgt}\\
&=\frac{M+1}{kM}\Biggl[ \sum_{l=1}^M {M\choose l}(-1)^{l+1}l F_0(n-l-1,1+a)
+\sum_{l=1}^{M} {M\choose l-1}(-1)^{l+1}(l -1)F_0(n-l-1,1+a)\Biggr]\nonumber\\
&\quad+\frac{M+1}{k}(-1)^{M+2}F_0(n-M-2,1+a)\label{chgt2}\\
&=\frac{1}{k} \sum_{l=1}^M {M+1\choose l}(-1)^{l+1}l F_0(n-l-1,1+a)
+\frac{M+1}{k}(-1)^{M+2}F_0(n-M-2,1+a)\quad \hbox{using } (\ref{basic-id})\nonumber\\
&=  \frac{1}{k}\sum_{l=1}^{M+1} {M+1\choose l}(-1)^{l+1}l F_0(n-l-1,1+a),\nonumber
\end{align}
\endgroup
where we perform a change of variable $l \to l-1$ in the second sum of (\ref{chgt}) to obtain (\ref{chgt2}).
This concludes the induction step and the proof.
\end{proof}
\section{Useful Formulas}
\label{app:useful-formulas}
It is known that  
\begin{equation}
\sum_{i=0}^n {n\choose i} (-1)^i P(i)=0,
\label{poly-result}
\end{equation}
for any polynomial $P$ of degree less than $n$ \cite{Ruiz96}.
\begin{lemma}
For any integer $J\geq 1$, $0\leq z_1+\cdots+z_J<1$, and $L\in \N$,
\begin{align}
\label{gen-formula}
\sum_{n_1\geq 0,\ldots,n_J\geq 0} \frac{(L+n_1+\cdots+n_J)!}{n_1!\cdots n_J!} z_1^{n_1}\cdots z_J^{n_J}=\frac{L!}{(1-z_1-\cdots-z_J)^{L+1}}.
\end{align}
\end{lemma}
\begin{proof} Fix $L\geq 0$. For $ 0\leq z<1$, we have
\[
\sum_{j\geq 0} \frac{(L+j)!}{j!}z^{j}= \frac{L!}{z^L} \sum_{i\geq L} {i\choose L} z^i=  \frac{L!}{z^L} \times \frac{z^L}{(1-z)^{L+1}}= \frac{L!}{(1-z)^{L+1}},
\]
which shows that (\ref{gen-formula}) holds when $J=1$. Assume that  (\ref{gen-formula}) is true for $J=1,\ldots,m$ and $0\leq z_1+\cdots+z_m<1$,
and let us show that it is true for $J=m+1$ and $0\leq z_1+\cdots+z_{m+1}<1$. Take $0\leq z_1+\cdots+z_{m+1}<1$. We have
\begin{align*}
\sum_{n_1\geq 0,\ldots,n_{m+1}\geq 0}\hspace{-2mm}& \frac{(L+n_1+\cdots+n_{m+1})!}{n_1!\cdots n_{m+1}!} z_1^{n_1}\cdots z_{m+1}^{n_{m+1}}\\
&=
\sum_{n_1\geq 0,\ldots,n_{m}\geq 0} \frac{z_1^{n_1}\cdots z_m^{n_m}}{n_1!\cdots n_{m}!} \sum_{n_{m+1}\geq 0}\hspace{-2mm}\frac{(L+n_1+\cdots+n_m+n_{m+1})!}{n_{m+1}!}z_{m+1}^{n_{m+1}}
\\
&=\sum_{n_1\geq 0,\ldots,n_{m}\geq 0}\hspace{-2mm}\frac{(L+n_1+\cdots+n_m)!}{n_1!\cdots n_{m}!} \left(\frac{z_1}{1-z_{m+1}}\right)^{n_1}\cdots
 \left(\frac{z_m}{1-z_{m+1}}\right)^{n_{m}}\\
 &=\frac{L!}{(1-z_1-\cdots -z_{m+1})^{L+1}},
\end{align*}
where the last two identities follow from the induction hypothesis with $J=1$ and $J=m$, respectively.
\end{proof}
Differentiating (\ref{gen-formula}) w.r.t. $z_J$ and multiplying both sides of the resulting equation by $z_J$ yields, for $0\leq z_1+\cdots+z_J<1$.
\begin{align}
\sum_{n_1\geq 0,\ldots,n_J\geq 0}\hspace{-4mm} \frac{(L+n_1+\cdots+n_J)!}{n_1!\cdots n_J!} z_1^{n_1}\cdots z_J^{n_J}n_J=
\frac{(L+1)!z_J}{(1-z_1-\cdots -z_J)^{L+2}}.
\label{app:diff}
\end{align}
In the following, assume  $0\leq z <\frac{1}{J}$ and $L\in \N$.
 Letting $z_j=z$ for $j=1,\ldots,J$ in both (\ref{gen-formula}) and (\ref{app:diff}) yields
\begin{align}
\label{app:formula-0}
\sum_{n_1\geq 0,\ldots,n_J\geq 0}\hspace{-3mm} \frac{(L+n_1+\cdots+n_J)!}{n_1!\cdots n_J!} z^{\sum\limits_{i=1}^Jn_i}(n_J-1)=
\frac{L!((J+L+1)z-1)}{(1-Jz)^{L+2}}.
\end{align}
Differentiating (\ref{app:diff}) w.r.t. $z_J$, multiplying both sides of the resulting equation by $z_J$, and then letting $z_j=z$ for all $j$ yields
\begin{align}
\label{app:diff2}
\sum_{n_1\geq 0,\ldots,n_J\geq 0}\hspace{-4mm} \frac{(L+n_1+\cdots+n_J)!}{n_1!\cdots n_J!}z^{\sum\limits_{i=1}^Jn_i} n_J^2=
\frac{(L+1)!((L+2-J)z+1)z}{(1-Jz)^{L+3}}.
\end{align}
Differentiating  (\ref{app:diff}) w.r.t. $z_I$ with $I\neq J$,   multiplying both sides of the resulting equation by $z_I $, then letting $z_j=z$ for all $j$ yields
\begin{align}
\label{app:diff-cross}
\sum_{n_1\geq 0,\ldots,n_J\geq 0} \frac{(L+n_1+\cdots+n_J)!}{n_1!\cdots n_J!} z^{n_1+\cdots+n_J}n_I n_J=
\frac{(L+2)!z^2}{(1-Jz)^{L+3}}.
\end{align}
\section{Expected Number of Stored Qubits at the Switch}
\label{app:EQ-long-proof}
Following is a proof of Proposition \ref{prop:EQlonger}.
\begin{proof}
For any vector $\vect{y}=(y_1~\ldots~y_{n-1})$, define
$|\vect y |^{(2)}\coloneqq\sum_{j=1}^{n-1} y^2_j$.

In Section 7, we prove that the system is stable if and only if $k>n$. Hence, assume from now on that $k>n$. 
In Section 8, we prove that $\E[|\vect Q|]<\infty$ when the system is stable.
Assume that $X$ is in steady state at time $t=1$ (which implies that it is in steady-state at any time $t>1$). By Lemma \ref{lemma:eq40}, Eq. (\ref{eq:40}) holds.
Take  $V(\vect{x})=\min\{|\vect{x}|^2,T^2\}$ for $T\geq  n-1$. Multiplying both sides of (\ref{eq:40}) by $k$, we get
\begin{align}
0&=(k-(n-1)) \sum_{\vect{x}\in R_0}  \pi(\vect{x})\left[\min\{(|\vect{x}|-(n-1))^2, T^2\}-\min\{|\vect{x}|^2, T^2\}\right]\nonumber\\
&\quad+(n-1) \sum_{\vect{x}\in R_0}  \pi(\vect{x})\left[\min\{(|\vect{x}|+1)^2, T^2\}-\min\{|\vect{x}|^2, T^2\}\right]\nonumber\\
&\quad+k\sum_{j=1}^{n-1}\sum_{\vect{x}\in R_j} \pi(\vect{x})\left[\min\{(|\vect{x}|+1)^2, T^2\}-\min\{|\vect{x}|^2, T^2\}\right]\nonumber\\
&= -2(k-(n-1))(n-1)\sum_{\vect{x}\in R_0\atop n-1\leq |\vect{x}|\leq T-1}  \pi(\vect{x}) |\vect{x}|
+(k-(n-1)) \sum_{\vect{x}\in R_0\atop T\leq  |\vect{x}|<T+n-1}  \pi(\vect{x})((|\vect{x}|-(n-1))^2 -T^2)\nonumber\\
&\quad+(k-(n-1))(n-1)^2 \sum_{\vect{x}\in R_0\atop n-1\leq |\vect{x}|\leq T-1} \pi(\vect{x})
+(n-1)  \sum_{\vect{x}\in R_0\atop T\geq |\vect{x}|+1}  \pi(\vect{x})(2|\vect{x}|+1)
+k\sum_{j=1}^{n-1}  \sum_{\vect{x}\in R_j\atop T\geq |\vect{x}|+1}  \pi(\vect{x})(2|\vect{x}|+1).
\label{eq:EQ41}
\end{align}
Define
\begin{align*}
A&\coloneqq \sum_{\vect{x}\in R_0}  \pi(\vect{x}) |\vect{x}|,
\qquad
B\coloneqq \sum_{j=1}^{n-1} \sum_{\vect{x}\in R_j}  \pi(\vect{x}) |\vect{x}|=\sum_{j=1}^{n-2} \sum_{\vect{x}\in R_j}  \pi(\vect{x}) |\vect{x}|,\\
C&\coloneqq \sum_{\vect{x}\in R_0}  \pi(\vect{x}),
\qquad\quad
D\coloneqq \sum_{j=1}^{n-1} \sum_{\vect{x}\in R_j}  \pi(\vect{x}),
\end{align*}
where the second equality in the definition of $B$ holds since $R_{n-1}=\{\vect{0}\}$. With these definitions  (\ref{eq:EQ41}) becomes
\begin{align}
0&=
-2(k-(n-1))(n-1)\left[A- \sum_{\vect{x}\in R_0\atop T\leq |\vect{x}|}  \pi(\vect{x}) |\vect{x}|\right]+(k-(n-1)) \sum_{\vect{x}\in R_0\atop T\leq  |\vect{x}|\leq T+n-1}  \pi(\vect{x})((|\vect{x}|-(n-1))^2 -T^2)\nonumber\\
&\quad+(k-(n-1))(n-1)^2\left[C-  \sum_{\vect{x}\in R_0\atop T\leq  |\vect{x}|} \pi(\vect{x})\right]
+(n-1) \left[2A- 2\sum_{\vect{x}\in R_0\atop T\leq |\vect{x}|} \pi(\vect{x})|\vect{x}|   +C-\sum_{\vect{x}\in R_0\atop T\leq |\vect{x}|} \pi(\vect{x} )\right]\nonumber\\
&\quad+k\left[ 2B - \sum_{j=1}^{n-1}  \sum_{\vect{x}\in R_j\atop T\leq  |\vect{x}|}  \pi(\vect{x})|\vect{x}|    +D- \sum_{j=1}^{n-1}  \sum_{\vect{x}\in R_j\atop T\leq  |\vect{x}|}  \pi(\vect{x})\right]\nonumber\\
&= -2(n-1)(k-1)A +2kB+(n-1)C +(n-1)[(k-(n-1))(n-1)+1]C+kD + g(T),
\label{eq:EQ70}
\end{align}
with
\begin{align}
\lefteqn{g(T)\coloneqq 
2(n-1)(k-n)  \sum_{\vect{x}\in R_0\atop T\leq  |\vect{x}|}  \pi(\vect{x}) |\vect{x}| -(n-1)[(k-(n-1))(n-1)-1] 
 \sum_{\vect{x}\in R_0\atop T\leq  |\vect{x}|}  \pi(\vect{x})}\nonumber\\
 &+(k-(n-1)) \sum_{\vect{x}\in R_0\atop T\leq  |\vect{x}|<T+n-1}  \pi(\vect{x})((|\vect{x}|-(n-1))^2 -T^2) -k\left[ \sum_{j=1}^{n-1}  \sum_{\vect{x}\in R_j\atop T\leq |\vect{x}|}  \pi(\vect{x})|\vect{x}| +\sum_{j=1}^{n-1}\sum_{\vect{x}\in R_j\atop T\leq |\vect{x}|} \pi(\vect{x})\right].
 \label{EQdef-g}
\end{align}
We know that (see proof of Proposition \ref{capacity-n}, Eq. (\ref{eq:43}))
\begin{equation}
\label{eq:EQ71}
C=\frac{k}{n(k-(n-1))},
\end{equation}
so that by Eq. (\ref{eq:42}),
\begin{equation}
\label{eq:EQ72}
D=1-\sum_{ \vect{x} \in R_0}  \pi(\vect{x})= 1-\frac{k}{n(k-(n-1))}=
\frac{(k-n)(n-1)}{n(k-(n-1))}.
\end{equation}
Introducing  (\ref{eq:EQ71}) and  (\ref{eq:EQ72}) into (\ref{eq:EQ70}) gives
\begin{equation}
2(k-n)(n-1)A-2kB  = k(n-1) +g(T).
\label{EQfirst-eqn}
\end{equation}
Take now $V(\vect{x})=\min\{|\vect{x}|^{(2)},T\}$ for $T\geq n-1$. Multiplying both sides of (\ref{eq:40}) by $k$, we get
\begin{align}
0&=\sum_{\vect{x}\in R_0}\pi(\vect{x})(k-(n-1)) \times
\left(\min\left\{|\vect{x} - {\bf 1}|^{(2)}, T\right\}-\min\left\{ |\vect{x}|^{(2)},T\right\}\right)\nonumber\\
&\quad+\sum_{\vect{x}\in R_0}\pi(\vect{x})
\sum_{l=1}^{n-1} \left(\min\left\{|\vect{x}|^{(2)}+2x_l+1, T\right\}-\min\left\{|\vect{x}|^{(2)},T\right\}\right)\nonumber\\
&\quad
+\sum_{j=1}^{n-1} \sum_{\vect{x}\in R_j}\pi(\vect{x})\sum_{l=1}^j \left(\frac{k-(n-1)-j}{j}\right)\times \left(\min\left\{1+|\vect{x}|^{(2)},T\right\}-
\min\left\{|\vect{x}|^{(2)},T\right\}\right)\nonumber\\
&\quad+
\sum_{j=1}^{n-1} \sum_{\vect{x}\in R_j}\pi(\vect{x}) \sum_{m=1}^{n-1}\ind_{\{x_m>0\}}\left(\min\left\{ |\vect{x}|^{(2)}+2x_m+1,T\right\}-\min\left\{|\vect{x}|^{(2)},T\right\}\right).
\label{eq:EQ300}
\end{align}
Note that the last summation in (\ref{eq:EQ300}) has $n-1-j$ terms as $\vect{x}=(x_1~\ldots~x_{n-1})\in R_j$ has $n-1-j$ non-zero entries by definition of the set $R_j$.
Let us consider separately the four terms in the r.h.s. of (\ref{eq:EQ300}). Call these terms I, II, III, and IV respectively.
We have
\begin{align*}
I&=(k-(n-1))\sum_{\vect{x}\in R_0\atop |\vect{x} |^{(2)}\leq T}\pi(\vect{x})  (-2|\vect{x}|+n-1)
+(k-(n-1))\sum_{\vect{x}\in R_0\atop |\vect{x}  -{\bf 1}|^{(2)} <T<|\vect{x} |^{(2)}}\pi(\vect{x}) \left( 
|\vect{x} - {\bf 1}|^{(2)}- T\right)\nonumber\\
&= -2(k-(n-1)) A+(k-(n-1))(n-1) \sum_{\vect{x}\in R_0}\pi(\vect{x})
-(k-(n-1))\sum_{\vect{x}\in R_0\atop T< |\vect{x} |^{(2)}}\pi(\vect{x})  (-2|\vect{x}|+n-1)\\
&\quad+(k-(n-1))\sum_{\vect{x}\in R_0\atop    |\vect{x}  -{\bf 1}|^{(2)}<T< |\vect{x} |^{(2)}}\pi(\vect{x}) \left(|\vect{x} - {\bf 1}|^{(2)}- T\right),\\
%
II&=\sum_{l=1}^n \sum_{\vect{x}\in R_0\atop  |\vect{x} |^{(2)}+2x_l+1\leq T}\pi(\vect{x})(2x_l+1)
+\sum_{l=1}^n \sum_{\vect{x}\in R_0\atop   |\vect{x} |^{(2)} <T<  |\vect{x} |^{(2)}+2x_l+1}\pi(\vect{x})
\left(T-|\vect{x} |^{(2)} \right)\nonumber\\
&=\sum_{l=1}^{n-1}\sum_{\vect{x}\in R_0} \pi(\vect{x})(2x_l+1)
-\sum_{l=1}^{n-1}\sum_{\vect{x}\in R_0\atop T< |\vect{x} |^{(2)}+2x_l+1} \pi(\vect{x})(2x_l+1)
+\sum_{l=1}^n \sum_{\vect{x}\in R_0 \atop|\vect{x} |^{(2)} <T<  |\vect{x} |^{(2)}+2x_l+1}\pi(\vect{x})
\left(T-  |\vect{x} |^{(2)}  \right)\nonumber\\
&=2A +(n-1)\sum_{\vect{x}\in R_0} \pi(\vect{x})
-\sum_{l=1}^{n-1} \sum_{\vect{x}\in R_0\atop T< |\vect{x} |^{(2)} +2 x_l+1} \pi(\vect{x})(2x_l+1)
+\sum_{l=1}^{n-1} \sum_{\vect{x}\in R_0\atop |\vect{x} |^{(2)} <T<  |\vect{x} |^{(2)}+2x_l+1}\pi(\vect{x})
\left(T-|\vect{x} |^{(2)}\right),\\
%
III&= \sum_{j=1}^{n-1}(k-(n-1-j))  \sum_{\vect{x}\in R_j \atop |\vect{x}|^{(2)}+1\leq T} \pi(\vect{x})  \nonumber\\
&= \sum_{j=1}^{n-1}(k-(n-1-j))  \sum_{\vect{x}\in R_j}\pi(\vect{x})  
-  \sum_{j=1}^{n-1}(k-(n-1-j))  \sum_{\vect{x}\in R_j\atop  |\vect{x}|^{(2)}+1> T} \pi(\vect{x}),\\
%
IV&=\sum_{j=1}^{n-1}\sum_{m=1}^{n-1}  \sum_{\vect{x}\in R_j \atop |\vect{x}|^{(2)}+2x_m+1\leq T} \pi(\vect{x}) \ind_{\{x_m>0\}}(2x_m+1)
+\sum_{j=1}^{n-1}\sum_{m=1}^{n-1}  \sum_{\vect{x}\in R_j \atop |\vect{x}|^{(2)} < T <|\vect{x}|^{(2)}+2x_m+1} \pi(\vect{x}) \ind_{\{x_m>0\}}\left(T-|\vect{x}|^{(2)}\right)\nonumber\\
&= \sum_{j=1}^{n-1} \sum_{m=1}^{n-1} \sum_{\vect{x}\in R_j} \ind_{\{x_m>0\}} \pi(\vect{x}) (2x_m+1) -
 \sum_{j=1}^{n-1} \sum_{m=1}^{n-1} \sum_{\vect{x}\in R_j\atop  T< |\vect{x}|^{(2)}+2x_m+1}\ind_{\{x_m>0\}} \pi(\vect{x}) (2x_m+1)\nonumber\\
&\quad+\sum_{j=1}^{n-1}\sum_{m=1}^{n-1}  \sum_{\vect{x}\in R_j \atop |\vect{x} |^{(2)} < T <|\vect{x}|^{(2)}+2x_m+1} \pi(\vect{x}) \ind_{\{x_m>0\}}
\left(T-|\vect{x}|^{(2)}\right)\nonumber\\
&=2 B +\sum_{j=1}^{n-1} (n-1-j) \sum_{\vect{x}\in R_j}  \pi(\vect{x}) 
- \sum_{j=1}^{n-1} \sum_{m=1}^{n-1} \sum_{\vect{x}\in R_j\atop  T< |\vect{x}|^{(2)}+2x_m+1}\ind_{\{x_m>0\}} \pi(\vect{x}) (2x_m+1)\nonumber\\
&\quad+\sum_{j=1}^{n-1} \sum_{m=1}^{n-1}  \sum_{\vect{x}\in R_j \atop |\vect{x} |^{(2)} < T <|\vect{x}|^{(2)}+2x_m+1} \pi(\vect{x}) \ind_{\{x_m>0\}}
\left(T-|\vect{x}|^{(2)}\right). 
\end{align*}
Hence,
\begin{align}
0&= I+II+III+IV\nonumber\\
&=-2(k-n)A +2B+(k-(n-1))(n-1) \sum_{\vect{x}\in R_0}\pi(\vect{x})   +(n-1)\sum_{\vect{x}\in R_0} \pi(\vect{x})  \nonumber\\
&\quad+\sum_{j=1}^{n-1}(k-(n-1-j))  \sum_{\vect{x}\in R_j}\pi(\vect{x})  +\sum_{j=1}^{n-1} (n-1-j)\sum_{\vect{x}\in R_j}  \pi(\vect{x}) +h(T)\nonumber\\
&=-2(k-n)A +2B+(n-1)(k-n+2)\sum_{\vect{x}\in R_0} \pi(\vect{x}) +k \sum_{j=1}^{n-1} \sum_{\vect{x}\in R_j}\pi(\vect{x})+h(T) \nonumber\\
&=-2(k-n)A +2B+\frac{2(n-1)k}{n} +h(T)
\label{EQsecond-eqn}
\end{align}
by using (\ref{eq:EQ71}) and (\ref{eq:EQ72}), where 
\begin{align}
h(T)&\coloneqq -(k-(n-1))\sum_{\vect{x}\in R_0\atop T< |\vect{x} |^{(2)}}\pi(\vect{x})  (-2|\vect{x}|+n-1)+(k-(n-1))\sum_{\vect{x}\in R_0\atop    |\vect{x}  -{\bf 1}|^{(2)}<T< |\vect{x} |^{(2)}}\pi(\vect{x}) 
\left(|\vect{x} - {\bf 1}|^{(2)}- T\right)\nonumber\\
&\quad-\sum_{l=1}^{n-1} \sum_{\vect{x}\in R_0\atop T<  |\vect{x} |^{(2)} +2 x_l+1} \pi(\vect{x})(2x_l+1)
+\sum_{l=1}^{n-1} \sum_{\vect{x}\in R_0\atop |\vect{x} |^{(2)} <T<  |\vect{x} |^{(2)}+2x_l+1}\pi(\vect{x})
\left(T-|\vect{x} |^{(2)}\right)\nonumber\\
&\quad-  \sum_{j=1}^{n-1}(k-(n-1-j))  \sum_{\vect{x}\in R_j\atop T< |\vect{x}|^{(2)}+1} \pi(\vect{x})
- \sum_{j=1}^{n-1} \sum_{m=1}^{n-1} \sum_{\vect{x}\in R_j\atop  T< |\vect{x}|^{(2)}+2x_m+1}\ind_{\{x_m>0\}} \pi(\vect{x}) (2x_m+1)\nonumber\\
&\quad+\sum_{j=1}^{n-1} \sum_{m=1}^{n-1}  \sum_{\vect{x}\in R_j \atop |\vect{x} |^{(2)} < T <|\vect{x}|^{(2)}+2x_m+1} \pi(\vect{x}) \ind_{\{x_m>0\}}
\left(T-|\vect{x}|^{(2)}\right).
\label{def:EQh}
\end{align}
Eqns (\ref{EQfirst-eqn}) and (\ref{EQsecond-eqn}) give
the following set of linear equations in the unknowns $A$ and $B$:
\begin{align*}
2(k-n)(n-1)A-2kB&=k(n-1)+g(T)\\
2(k-n)A- 2B&=\frac{2k(n-1)}{n}+h(T). 
\end{align*}
Since $k>n$ there is a unique solution given by 
\begin{align}
A&=\frac{k(n-1)(2k-n)+n(kh(T)-g(T))}{2n(k-n)(k-(n-1))}
\label{sol:EQA}\\
B&= 
\frac{k(n-1)(n-2)+n((n-1)h(T)-g(T))}{2n(k-(n-1))}.
\label{sol:EQB}
\end{align}
Let us show that $\lim_{T\to\infty} g(T)=\lim_{T\to\infty}h(T)=0$.
Fix $j=1,\ldots,n-1$. Since
 $\sum_{\vect{x}\in R_j\atop T\leq |\vect{x}|} \pi(\vect{x})|\vect{x}|=\sum_{l=T}^\infty \sum_{\vect{x}\in R_j\atop |\vect{x}|=l} \pi(\vect{x})|\vect{x}|$
and since the series $\sum_{\vect{x}\in R_j\atop T\leq |\vect{x}|} \pi(\vect{x})|\vect{x}|$ is finite (bounded by $\E[|\vect Q|]$, which is finite due to the chain's stability), then necessarily $\lim_{T\to\infty}\sum_{l=T}^\infty \sum_{\vect{x}\in R_j\atop |\vect{x}|=l} \pi(\vect{x}) |\vect{x}|=0$,  which in turn implies that $\lim_{T\to\infty}\sum_{\vect{x}\in R_j\atop T\leq |\vect{x}|} \pi(\vect{x})|\vect{x}|=0$. This latter limit combined with the two-sided inequalities $0\leq \sum_{\vect{x}\in R_j\atop T\leq |\vect{x}|} \pi(\vect{x}) \leq \sum_{\vect{x}\in R_j\atop T\leq |\vect{x}|} \pi(\vect{x})|\vect{x}|$  imply that 
$\lim_{T\to\infty}\sum_{\vect{x}\in R_j\atop T\leq |\vect{x}|} \pi(\vect{x})=0$. Hence, cf. (\ref{EQdef-g}),
\[
\lim_{T\to\infty} g(T)= (k-(n-1))\times \lim_{T\to\infty}\sum_{\vect{x}\in R_0\atop T\leq  |\vect{x}|<T+n-1}  \pi(\vect{x})( (|\vect{x}|-(n-1))^2 -T^2).
\]
Let us evaluate this limit. We have
\begin{align}
0&\geq \sum_{\vect{x}\in R_0\atop T\leq  |\vect{x}|<T+n-1}  \pi(\vect{x})  (( |\vect{x}|-(n-1))^2 -T^2)\nonumber\\
&= \sum_{\vect{x}\in R_0\atop T\leq  |\vect{x}|<T+n-1}  \pi(\vect{x}) (|\vect{x}|^2 -T^2) 
-2(n-1)\times \sum_{\vect{x}\in R_0\atop T\leq  |\vect{x}|<T+n-1}  \pi(\vect{x}) |\vect{x}|+
(n-1)^2 \times\sum_{\vect{x}\in R_0\atop T\leq  |\vect{x}|<T+n-1}  \pi(\vect{x})\nonumber\\
&\geq  -2(n-1)\times \sum_{\vect{x}\in R_0\atop T\leq  |\vect{x}|<T+n-1}  \pi(\vect{x}) |\vect{x}|+
(n-1)^2 \times\sum_{\vect{x}\in R_0\atop T\leq  |\vect{x}|<T+n-1}  \pi(\vect{x}).
\label{inq:EQ100}
\end{align}
Now, the inequalities  $0\leq \sum_{\vect{x}\in R_0\atop T\leq  |\vect{x}|<T+n-1} \pi(\vect{x}) \leq \sum_{\vect{x}\in R_0\atop T\leq |\vect{x}|} \pi(\vect{x})$ and 
$0\leq \sum_{\vect{x}\in R_0\atop T\leq  |\vect{x}|<T+n-1}  \pi(\vect{x}) |\vect{x}|\leq \sum_{\vect{x}\in R_0\atop T\leq  |\vect{x}|}  \pi(\vect{x}) |\vect{x}|$
combined with the limits $\lim_{T\to\infty}\sum_{\vect{x}\in R_j\atop T\leq |\vect{x}|} \pi(\vect{x})=0$ and $\lim_{T\to\infty}\sum_{\vect{x}\in R_j\atop T\leq |\vect{x}|} \pi(\vect{x})|\vect{x}|=0$ shown above (the latter limit holds since $\E[|\vect Q|]<\infty$), prove  
that 
\begin{align*}
\lim_{T\to\infty} \sum_{\vect{x}\in R_0\atop T\leq  |\vect{x}|<T+n-1} \pi(\vect{x})=\lim_{T\to\infty} \sum_{\vect{x}\in R_0\atop T\leq  |\vect{x}|<T+n-1} \pi(\vect{x})|\vect{x}|=0.
\end{align*} 
We then conclude from (\ref{inq:EQ100}) that 
\begin{equation}
\label{EQlim-g}
\lim_{T\to\infty} g(T)= 0.
\end{equation}
The proof that $\lim_{T\to\infty} h(T)= 0$ is similar; we compute this limit next. 
%
%
Throughout we will use the fact that for any mapping $\phi:\{0,1,\ldots\}^{n-1}\to [0,\infty)$ the limits
\begin{equation}
\label{EQfundamental}
\lim_{T\to \infty} \sum_{\vect{x}\in R_j\atop T<\phi(\vect{x})}\pi(\vect{x}) =0 \,\,\hbox{ and }\,\,
\lim_{T\to \infty} \sum_{\vect{x}\in R_j\atop T<\phi(\vect{x})}\pi(\vect{x})|\vect{x}|=0,
\quad \forall  j=0,1,\ldots,n-1,
\end{equation}
hold as  the series  $\sum_{\vect{x}\in R}\pi(\vect{x})$ and  $\sum_{\vect{x}\in R}\pi(\vect{x})|\vect{x}|=\E[|\vect Q|]$ are finite, respectively.\\

By applying (\ref{EQfundamental})  to (\ref{def:EQh}) we immediately conclude that the first, third, fifth and sixth terms in the r.h.s of (\ref{def:EQh}) go to zero as $T\to\infty$.  Hence,
\begin{align}
\lim_{T\to\infty}h(T)&=\lim_{T\to \infty}(k-(n-1))\sum_{\vect{x}\in R_0\atop    |\vect{x}  -{\bf 1}|^{(2)}<T< |\vect{x} |^{(2)}}\pi(\vect{x}) 
\left(|\vect{x} - {\bf 1}|^{(2)}- T\right)\nonumber\\
&\quad+\lim_{T\to\infty}\sum_{l=1}^{n-1} \sum_{\vect{x}\in R_0\atop |\vect{x} |^{(2)} <T<  |\vect{x} |^{(2)}+2x_l+1}\pi(\vect{x})
\left(T-|\vect{x} |^{(2)}\right)\nonumber\\
&\quad+\lim_{T\to \infty} \sum_{j=1}^{n-1} \sum_{m=1}^{n-1}  \sum_{\vect{x}\in R_j \atop |\vect{x} |^{(2)} < T <|\vect{x}|^{(2)}+2x_m+1} \pi(\vect{x}) \ind_{\{x_m>0\}}
\left(T-|\vect{x}|^{(2)}\right).
\label{app:EQlim-h}
\end{align}
Let us show that the three limits in the r.h.s. of (\ref{app:EQlim-h}) are equal to zero.
We have
\begin{align*}
0&\geq \sum_{\vect{x}\in R_0\atop    |\vect{x}  -{\bf 1}|^{(2)}<T< |\vect{x} |^{(2)}}\pi(\vect{x})
\left(|\vect{x} - {\bf 1}|^{(2)}- T\right)
=  \sum_{\vect{x}\in R_0\atop  T<   |\vect{x} |^{(2)} <T-(n-1) -2|\vect{x}|}\pi(\vect{x}) \left(|\vect{x}|^{(2)} - 2|\vect{x}| +n-1-T\right)\nonumber\\
&\geq
 -2 \sum_{\vect{x}\in R_0\atop    T<   |\vect{x} |^{(2)} <T-(n-1) -2|\vect{x}|}\pi(\vect{x}) |\vect{x}| 
 +(n-1)\sum_{\vect{x}\in R_0\atop    T<   |\vect{x} |^{(2)} <T-(n-1) -2|\vect{x}|} \pi(\vect{x}).
\end{align*}
Now, 
\[
0\leq \sum_{\vect{x}\in R_0\atop    T<   |\vect{x} |^{(2)} <T-(n-1) -2|\vect{x}|}\pi(\vect{x}) |\vect{x}| \leq 
 \sum_{\vect{x}\in R_0\atop    T<   |\vect{x} |^{(2)}}\pi(\vect{x}) |\vect{x}|\to 0 \quad \hbox{as } T\to \infty
 \]
and
\[
0\leq \sum_{\vect{x}\in R_0\atop    T<   |\vect{x} |^{(2)} <T-(n-1) -2|\vect{x}|}\pi(\vect{x}) \leq 
 \sum_{\vect{x}\in R_0\atop    T<   |\vect{x} |^{(2)}}\pi(\vect{x})\to 0 \quad \hbox{as } T\to \infty,
 \]
 where the limits hold from (\ref{EQfundamental}).
Therefore, 
\begin{equation}
\lim_{T\to\infty}\sum_{\vect{x}\in R_0\atop    |\vect{x}  -{\bf 1}|^{(2)}<T< |\vect{x} |^{(2)}}\pi(\vect{x}) \left(|\vect{x} - {\bf 1}|^{(2)}- T\right)=0.
\label{EQlim2}
\end{equation}
On the other hand,
\begin{align}
\lefteqn{0\leq \sum_{l=1}^{n-1} \sum_{\vect{x}\in R_0\atop |\vect{x} |^{(2)} <T<  |\vect{x} |^{(2)}+2x_l+1}\pi(\vect{x})\left(T-|\vect{x} |^{(2)}\right)
\leq 
\sum_{l=1}^{n-1} \sum_{\vect{x}\in R_0\atop |\vect{x} |^{(2)} <T<  |\vect{x} |^{(2)}+2x_l+1}\pi(\vect{x})(2x_l+1)}\nonumber\\
&\leq
2 \sum_{l=1}^{n-1} \sum_{\vect{x}\in R_0\atop T<  |\vect{x} |^{(2)}}  \pi(\vect{x})|\vect{x}|+(n-1)
\sum_{l=1}^{n-1} \sum_{\vect{x}\in R_0\atop T<  |\vect{x} |^{(2)}}\pi(\vect{x})\to 0 \quad \hbox{as } T\to\infty, 
\label{EQlim3}
\end{align}
by using (\ref{EQfundamental}).  Finally, for $j=1,\ldots,n-1$,
\begin{align}
0\leq \sum_{m=1}^{n-1}  \sum_{\vect{x}\in R_j \atop |\vect{x} |^{(2)} < T <|\vect{x}|^{(2)}+2x_m+1} \pi(\vect{x}) \ind_{\{x_m>0\}}
\left(T-|\vect{x}|^{(2)}\right)
\leq \sum_{m=1}^{n-1}  \sum_{\vect{x}\in R_j \atop |\vect{x} |^{(2)} < T <|\vect{x}|^{(2)}+2x_m+1} \pi(\vect{x})
\left(T-|\vect{x}|^{(2)}\right) 
\label{EQlim4}
\end{align}
which goes to zero as $T\to \infty$ from (\ref{EQlim3}). This proves that $\lim_{T\to\infty} h(T)=0$. 

We are now ready to conclude the proof.
Fix $\epsilon>0$. We can then find $T_\epsilon$ such that 
\[
\frac{2(n-1)h(T)-g(T)}{2(k-(n-1))}<\epsilon, \quad \frac{2(n-1)h(T)-g(T)}{2(k-(n-1))}<\epsilon,
\]
for all $T\geq T_\epsilon$. 
Take $T=T_\epsilon$ in the definition of the functions $V(\vect{x})=\min(|\vect{x}|^2,T^2)$ and  $V(\vect{x})=\min(\vect{x}^{(2)},T)$.
This yields
\begin{align}
\left|A-\frac{k(n-1)(2k-n)}{2n(k-n)(k-(n-1))}\right|&<\epsilon \label{sol:EQA2}\\
\left|B- \frac{k(n-1)(n-2)}{2n(k-(n-1))}\right|&<\epsilon. \label{sol:EQB2}
\end{align}
Since (\ref{sol:EQA2})-(\ref{sol:EQB2}) hold for any $\epsilon>0$, we conclude that  
\[
A=\frac{k(n-1)(2k-n)}{2n(k-n)(k-(n-1))},\quad  B=\frac{k(n-1)(n-2)}{2n(k-(n-1))},
\]
so that
\[
\E[|\vect Q|]=A+B=\frac{(n-1)k}{2(k-n)},
\]
which concludes the proof of Proposition \ref{prop:EQlonger}.
\end{proof}
\section{Stability Analysis}
\subsection{Proof of Lemma \ref{lem:sum-f}}
\label{app:proof-lemSumf}
\begin{proof}
For $j=1,\ldots,n-2$,  and for $j=n-1$ and $g(i)=1$ for all $i\in \N$, the mapping $f_j(\vect r)$ is symmetric w.r.t. $r_1,\ldots,r_j$.
Therefore, since  there are ${j\choose l}$ ways of having $l$ components equal to $1$ among $j$ components, we have 
\begin{equation}
\label{lem-j}
\sum_{\vect r\in E_j} f_j(\vect r)= \sum_{l=1}^{j}{j\choose l} \sum_{r_{l+1}\geq 2, \ldots, r_j \geq 2\atop r_{j+1}\geq 0, \ldots, r_{n-1}\geq 0} f_j(1,\ldots,1,r_{l+1},\ldots, r_{n-1}).
\end{equation}
Let us calculate the inner sum in (\ref{lem-j}). We have
\begin{align*}
\sum_{r_{l+1}\geq 2, \ldots, r_j \geq 2\atop r_{j+1}\geq 0, \ldots, r_{n-1}\geq 0}\hspace{-2mm} f_j(1,\ldots,1,r_{l+1},\ldots, r_{n-1})
&= 
\frac{l}{k^l}\hspace{-2mm}\sum_{r_{l+1}\geq 2, \ldots, r_j \geq 2\atop r_{j+1}\geq 0, \ldots, r_{n-1}\geq 0} \hspace{-3mm}
\frac{(l-1+r_{l+1}+\cdots +r_{n-1})!}{r_{l+1}!\ldots  r_{n-1}!}\left(\frac{1}{k}\right)^{\sum_{m=l+1}^{n-1} r_m} g(r_{n-1}),
\label{int:Gamma-j}
\end{align*}
so that
\begin{align*}
\sum_{\vect r\in E_j} f_j(\vect r) &=\sum_{l=1}^{j}{j\choose l}
\frac{l}{k^l} \sum_{r_{l+1}\geq 2, \ldots, r_j \geq 2\atop r_{j+1}\geq 0, \ldots, r_{n-1}\geq 0} 
\frac{(l-1+N_{l+1})!}{r_{l+1}!\ldots  r_{n-1}!}\ \left(\frac{1}{k}\right)^{N_{l+1}} g(r_{n-1}),
\end{align*}
which concludes the proof.
\end{proof}
\subsection{Proof of Proposition \ref{prop:drift}}
\label{app:proof-drift}
\begin{proof}
Define $V_1(\vect x)=\sum_{i=1}^{n-1}x_i^2$ and $V_2(\vect x)=\sum_{1\leq i<l\leq n-1}x_i x_l$, so that 
\begin{equation}
V(\vect x)=V_1(\vect x)+bV_2(\vect x).
\label{def:V}
\end{equation}
It will be convenient to consider separately the drifts $\Delta_{V_1}(\vect x)$ and $\Delta_{V_2}(\vect x)$ by
noting that $\Delta_{V}(\vect x) = \Delta_{V_1}(\vect x)+b \Delta_{V_2}(\vect x)$.

\subsubsection{Proof of (\ref{prop:delta-1})}
Take  $\vect x=(x_1~\ldots~x_{n-1})\in S-S^{\star}$. We have
\begin{align}
\Delta_{V_1}(\vect x)&=
\frac{k-n+1}{k} (V_1(\vect x -{\bf 1})- V_1(\vect x))+\frac{1}{k}\sum_{l=1}^{n-1}(V_1(\vect x+\vect e_l)-V_1(\vect x))\nonumber\\
&=
\frac{k-n+1}{k} (- 2|\vect x|+n-1)+\frac{1}{k}( 2|\vect x|+n-1)\nonumber\\
&=-2 \left(\frac{k-n}{k} \right)|\vect x|+(n-1)\left(\frac{k-n+2}{k}\right).
\label{foster1-V1}
\end{align}
It is a simple exercise to show that 
\begin{eqnarray*}
V_2(\vect x -{\bf 1})- V_2(\vect x)&=&-(n-2)|\vect x|+\frac{(n-1)(n-2)}{2}\\
V_2(\vect x+\vect e_l)-V_2(\vect x)&=& |\vect x|-x_l. 
\end{eqnarray*}
Hence,
\begin{align}
\Delta_{V_2}(\vect x)&=\frac{k-n+1}{k} (V_2(\vect x -{\bf 1})- V_2(\vect x))+\frac{1}{k}\sum_{l=1}^{n-1}(V_2(\vect x+\vect e_l)-V_2(\vect x))\nonumber\\
&=- \frac{(n-2)(k-n)}{k} |\vect x|+\frac{(k-n+1)(n-1)(n-2)}{2k}.
\label{foster1-V2}
\end{align}
From (\ref{foster1-V1}) and (\ref{foster1-V2}) we obtain that $k\Delta_{V}(\vect x) = k(\Delta_{V_1}(\vect x)+b \Delta_{V_2}(\vect x))$ is given by the r.h.s. of 
 (\ref{prop:delta-1}).
 \subsubsection{Proof of (\ref{prop:delta-2})}

Throughout $j\in \{1,\ldots,n-2\}$ is fixed and $\vect x_j'=(1~\ldots~1~x_{j+1}~\ldots~x_{n-1})$, $x_i\geq 2$ for $i=j+1,\ldots,n-1$, is also fixed.  
Define 
\[
\vect y(\vect r)=(r_1~\ldots~ r_j~ x_{j+1}-1+r_{j+1}~ \ldots~ x_{n-1}-1+r_{n-1}),
\]
with $\vect r=(r_1~\ldots~r_{n-1})\in E_j$ as defined in (\ref{def:Ej}).
We break the proof of (\ref{prop:delta-2}) in two parts, one for $V_1$ and the other for $V_2$.

\textit{Calculation of $\Delta_{V_1}(\vect x_j')$}:
By (\ref{prob-trans1})-(\ref{prob-trans3}), Lemma \ref{lem:hat-q} and the definition of $V_1$, we have 
\begin{align}
k\times \Delta_{V_1}(\vect x_j')&=(k-n+1)\sum_{ \vect r \in E_j} q^{\star}(\vect x_j' -{\bf 1}, \vect y(\vect r))
(V_1(\vect y(\vect r))-V_1(\vect x_j')) 
+\sum_{l=1}^{n-1}(V_1(\vect x_j'+\vect e_l)-V_1(\vect x))\nonumber\\
&=(j+1){k-n+j+1\choose k-n}\sum_{ \vect r \in E_j}   \sum_{l=1}^j \ind_{\{r_l=1\}} \frac{(|\vect{r}|-1)!}{r_1! \cdots r_{n-1}!}
 \left(\frac{1}{k}\right)^{|\vect{r}|}
(V_1(\vect y(\vect r))-V_1(\vect x_j')) \nonumber\\
&\qquad+2\sum_{i=j+1}^{n-1} x_i +n-1+2j,
\label{foster2}
\end{align}
where we have used the identity 
\[(k-n+1){k-n+j+1\choose k-n+1}=(j+1) {k-n+j+1\choose k-n}.\]
Since
\[
V_1(\vect y(\vect r))-V_1(\vect x_j')=2 \sum_{i=j+1}^{n-1} x_i(r_i-1)+\sum_{l=1}^{n-1}r^2_l -2\sum_{l=j+1}^{n-1} r_l +n-1-2j,
\]
elementary algebra from (\ref{foster2}) yields 
\begin{equation}
\label{Delta1-j}
k\times \Delta_{V_1}(\vect x_j')=2 \sum_{i=j+1}^{n-1} x_i\left[(j+1){k-n+j+1 \choose k-n}\Gamma_{j,i}+1\right]+\gamma_j,
\end{equation}
where, for $i=j+1,\ldots,n-1$,
\begin{align}
\Gamma_{j,i}&\coloneqq \sum_{\vect r\in E_j} \sum_{l=1}^j \ind_{\{r_l=1\}}
\frac{(|\vect{r}|-1) !}{r_1!\ldots  r_{n-1}!}  \left(\frac{1}{k}\right)^{|\vect{r}|} (r_i-1), \quad \text{and}
\label{def-Gamma-j}\\
\gamma_j&\coloneqq (j+1) {k-n+j+1\choose k-n} \sum_{\vect r\in E_j} \sum_{l=1}^j \ind_{\{r_l=1\}}
\frac{(|\vect{r}|-1)!}{r_1!\cdots  r_{n-1}!}\left(\frac{1}{k}\right)^{|\vect{r}|}
\left(\sum_{l=1}^{n-1}r^2_l -2\sum_{l=j+1}^{n-1} r_l +n-1-2j\right)\nonumber\\
&\qquad+n-1+2j.
\label{def-gamma-j}
\end{align}
It is easily seen from the definition of the set  $E_j$ in (\ref{def:Ej}) and the symmetry of the summand in (\ref{def-Gamma-j})
w.r.t. $r_{j+1},\ldots,r_{n-1}$,  that $\Gamma_{j,j+1}=\cdots=\Gamma_{j,n-1}$.
Define $\Gamma_j=\Gamma_{j,n-1}$. With this, (\ref{Delta1-j}) becomes
\begin{equation}
\label{Delta1-j-2}
k\times \Delta_{V_1}(\vect x_j')=2T_j\sum_{i=j+1}^{n-1} x_i+\gamma_j,
\end{equation}
where
\begin{equation}
T_j\coloneqq(j+1){k-n+j+1 \choose k-n}\Gamma_j+1.
\label{def:Tj}
\end{equation}
Let us further investigate $T_j$. 
We have (cf. (\ref{def-Gamma-j}) with $i=n-1$)
\begin{equation}
\label{eq-1:Gamma-j}
\Gamma_j= \sum_{\vect r \in E_j}  f_j(\vect r),
\end{equation}
where
\[
f_j(\vect r)\coloneqq \sum_{m=1}^j \ind_{\{r_m=1\}} \times 
\frac{(|\vect{r}|-1) !}{r_1!\ldots  r_{n-1}!} \times \left(\frac{1}{k}\right)^{|\vect{r}|}\times (r_{n-1}-1).
\]
By using the symmetry of $f_j$ wrt $r_1,\ldots,r_j$ 
we may rewrite $\Gamma_j$ as
\begin{equation}
\label{Gamma-j-alt}
\Gamma_j= \sum_{l=1}^{j}{j\choose l} \sum_{r_{l+1}\geq 2, \ldots, r_j \geq 2\atop r_{j+1}\geq 0, \ldots, r_{n-1}\geq 0} f_j(1,\ldots,1,r_{l+1},\ldots, r_{n-1}).
\end{equation}
By convention the set $\{r_{l+1}\geq 2, \ldots, r_j \geq 2\}$ is empty when $l=j$. In particular,
\begin{equation}
\label{Gamma-1}
\Gamma_1= \sum_{ r_2\geq 0,\ldots,r_{n-1}\geq 0}  f_1(1,r_2,,\ldots, r_{n-1}).
\end{equation}
Let us calculate the inner sum in (\ref{Gamma-j-alt}). 
For the following calculations, define $N_{i} \coloneqq r_i + \dots + r_{n-1}$.
We have
\begin{align*}
\sum_{r_{l+1}\geq 2, \ldots, r_j \geq 2\atop r_{j+1}\geq 0, \ldots, r_{n-1}\geq 0} f_j(1,\ldots,1,r_{l+1},\ldots, r_{n-1})
& =
l\left( \frac{1}{k}\right)^l \sum_{r_{l+1}\geq 2, \ldots, r_j \geq 2\atop r_{j+1}\geq 0, \ldots, r_{n-1}\geq 0} 
\frac{(l-1+N_{l+1})!}{r_{l+1}!\ldots  r_{n-1}!} \left(\frac{1}{k}\right)^{\sum_{m=l+1}^{n-1} r_m}  (r_{n-1}-1),
\end{align*}
for $l=1,\ldots,j$, so that
\begin{align}
\Gamma_j&=\sum_{l=1}^{j}{j\choose l}
\frac{l}{k^l} \sum_{r_{l+1}\geq 2, \ldots, r_j \geq 2\atop r_{j+1}\geq 0, \ldots, r_{n-1}\geq 0}
\frac{(l-1+N_{l+1})!}{r_{l+1}!\ldots  r_{n-1}!} \left(\frac{1}{k}\right)^{\sum_{m=l+1}^{n-1} r_m} 
(r_{n-1}-1)\nonumber\\
&=\sum_{l=1}^{j}{j\choose l} \frac{l}{k^l}F_{j-l}(r_{n-1}-1;n-l-1,l-1),
\label{Gamma-j-2}
\end{align}
where $F_m(g;J,L)$ is defined in (\ref{def-Fm}).
\begin{lemma}
\label{lem:T}
For $j=1,\ldots,n-2$,
\begin{equation}
T_j=-(k-n)+(k-n+1)\sum_{i=2}^{j+1}\frac{1}{k-n+i}.
\label{lem:Tj}
\end{equation}
\end{lemma}
The proof is given in Appendix \ref{app:Tj}.

Substituting (\ref{lem:Tj}) into  (\ref{Delta1-j-2}) yields
\begin{equation}
k\times \Delta_{V_1}(\vect x_j')=2 \left(-(k-n)+(k-n+1)\sum_{i=2}^{j+1}\frac{1}{k-n+i}\right)\sum_{i=j+1}^{n-1}x_i + \gamma_j.
\label{Delta-V1-foster2}
\end{equation}
Next, we compute the second portion of the drift for $\vect x_j'$.

\textit{Calculation of $\Delta_{V_2}(\vect x_j')$}:
Similarly to (\ref{foster2}), we find
\begin{align}
k\times \Delta_{V_2}(\vect x_j')&=(j+1){k-n+j+1\choose k-n}\sum_{ \vect r \in E_j}   \sum_{l=1}^j \ind_{\{r_l=1\}} \frac{(|\vect{r}|-1)!}{r_1! \cdots r_{n-1}!}
 \left(\frac{1}{k}\right)^{|\vect{r}|}
(V_2(\vect y(\vect r))-V_2(\vect x_j'))\nonumber\\
&\qquad +\sum_{l=1}^{n-1}(V_2(\vect x_j'+\vect e_l)-V_2(\vect x_j')).
\label{foster2-V2}
\end{align}
We have
\[
V_2(\vect x_j'+\vect e_l)-V_2(\vect x_j')= j+\sum_{i=j+1}^{n-1} x_i - x_l,
\]
so that
\begin{align}
 \sum_{l=1}^{n-1}(V_2(\vect x_j'+\vect e_l)-V_2(\vect x_j'))&= (n-1)j + (n-1) \sum_{i=j+1}^{n-1} x_i-\sum_{l=1}^{n-1}x_l\nonumber\\
 &=(n-1)j + (n-1) \sum_{i=j+1}^{n-1} x_i- \left(j + \sum_{l=j+1}^{n-1} x_l\right)\nonumber\\
 &= (n-2) \sum_{i=j+1}^{n-1} x_i +(n-2)j.
 \label{second-term-V2}
\end{align}
Let us now find  $V_2(\vect x_j')$ and $V_2(\vect y(\vect r))$. We have (recall that $x_1=\cdots=x_j=1$)
\begin{align}
V_2(\vect x_j')&= \sum_{i=1}^j \sum_{l=i+1}^{n-1} x_l + \sum_{i=j+1}^{n-2}x_i \sum_{l=i+1}^{n-1} x_l\nonumber\\
&=  \sum_{i=1}^j\left(j-i+\sum_{l=j+1}^{n-1}x_l\right)+ \sum_{j+1\leq i<l\leq n-1} x_ix_l\nonumber \\
&=\frac{j^2-j}{2}+j \sum_{l=j+1}^{n-1} x_l+\sum_{j+1\leq i<l\leq n-1} x_ix_l .
\label{V2-eq1}
\end{align}
On the other hand,
\begin{align*}
V_2(\vect y(\vect r))&=\sum_{i=1}^j r_i \left[\sum_{l=i+1}^j r_l + \sum_{l=j+1}^{n-1} (x_l-1+r_l)\right]
+\sum_{i=j+1}^{n-2}(x_i-1+r_i)\sum_{l=i+1}^{n-1}(x_l-1+r_l)\nonumber\\
&= \alpha_j(\vect r)+\hspace{-3mm}\sum_{j+1\leq i<l\leq n-1}\hspace{-3mm}x_ix_l +\sum_{i=j+1}^{n-1}x_i\left[\sum_{l=1}^{n-1}r_l -r_i -(n-j-2)\right],
\end{align*}
after elementary algebra, with 
\begin{equation}
\alpha_j(\vect r)\coloneqq\sum_{i=1}^j r_i\left[\sum_{l=i+1}^{n-1} r_l -(n-j-1)\right]+\sum_{j+1\leq i<l\leq n-1} (r_i-1)(r_l-1).
\label{def-alpha-j}
\end{equation}
Hence,
\begin{equation}
\label{diff-V2}
V_2(\vect y(\vect r))-V_2(\vect x_j')=\sum_{i=j+1}^{n-1}x_i\left[\sum_{l=1}^{n-1}r_l -r_i -(n-2)\right]+\alpha_j(\vect r)-\frac{j^2-j}{2}.
\end{equation}
Introducing (\ref{second-term-V2}) and (\ref{diff-V2}) into (\ref{foster2-V2}), we obtain
\begin{equation}
k\times \Delta_{V_2}(\vect x_j')=\sum_{i=j+1}^{n-1}x_i\left[n-2+ (j+1){k-n+j+1\choose k-n} \sum_{ \vect r \in E_j}   f_{j,i}(\vect r)  \right]+\beta_j,
\label{V2-eq3}
\end{equation}
where for $i=j+1,\ldots,n-1$,
\begin{align}
f_{j,i}(\vect r)\coloneqq \sum_{l=1}^j \ind_{\{r_l=1\}} \frac{(|\vect{r}|-1)!}{r_1! \cdots r_{n-1}!} \left(\frac{1}{k}\right)^{|\vect{r}|} \left(\sum_{l=1}^{n-1}r_l -r_i-(n-2)\right),
\label{def-f-j-i}
\end{align}
and
\begin{align}
\beta_j&\coloneqq (j+1){k-n+j+1\choose k-n}\sum_{ \vect r \in E_j}   \sum_{l=1}^j \ind_{\{r_l=1\}} \frac{(|\vect{r}|-1)!}{r_1! \cdots r_{n-1}!}
 \left(\frac{1}{k}\right)^{|\vect{r}|}
 \left(\alpha_j (\vect r)-\frac{j^2-j}{2}\right) +(n-2)j.
 \label{def-beta-j}
 \end{align}
The same argument used to derive  $\sum_{\vect r\in E_j}f_{j}$  (see (\ref{eq-1:Gamma-j}) and (\ref{Gamma-j-alt})) applies to $\sum_{\vect r\in E_j}f_{j,i}$, and again letting $N_{i} \coloneqq r_i + \dots + r_{n-1}$, we obtain
 \begin{align}
\sum_{\vect r\in E_j} f_{j,i}&=\sum_{l=1}^j {j\choose l} \sum_{r_{l+1}\geq 2, \ldots, r_j\geq 2 \atop r_{j+1}\geq 0,\ldots, r_{n-1}\geq 0}
f_{j,i}(1,\ldots,1,r_{l+1},\ldots,r_{n-1})\nonumber \\
&=\sum_{l=1}^j {j\choose l} l\left(\frac{1}{k}\right)^l   \sum_{r_{l+1}\geq 2, \ldots, r_j\geq 2 \atop r_{j+1}\geq 0,\ldots, r_{n-1}\geq 0}
 \frac{(l-1+N_{l+1})!}{r_{l+1}! \cdots   r_{n-1}!} \left(\frac{1}{k}\right)^{N_{l+1}}
 \left(l+N_{l+1} -r_i -(n-2)\right)\nonumber \\
 &=\sum_{l=1}^j {j\choose l} l\left(\frac{1}{k}\right)^l   \sum_{r_{l+1}\geq 2, \ldots, r_j\geq 2 \atop r_{j+1}\geq 0,\ldots, r_{n-1}\geq 0}
 \frac{(l+N_{l+1})!}{r_{l+1}! \cdots   r_{n-1}!} \left(\frac{1}{k}\right)^{N_{l+1}} \nonumber \\
 &\qquad-\sum_{l=1}^j {j\choose l} l\left(\frac{1}{k}\right)^l   \sum_{r_{l+1}\geq 2, \ldots, r_j\geq 2 \atop r_{j+1}\geq 0,\ldots, r_{n-1}\geq 0}
 \frac{(l-1+N_{l+1})!}{r_{l+1}! \cdots   r_{n-1}!} \left(\frac{1}{k}\right)^{N_{l+1}}r_i  
 \label{V2-eq2}\\
&\qquad-(n-2)\sum_{l=1}^j {j\choose l} l\left(\frac{1}{k}\right)^l   \sum_{r_{l+1}\geq 2, \ldots, r_j\geq 2 \atop r_{j+1}\geq 0,\ldots, r_{n-1}\geq 0}
 \frac{(l-1+N_{l+1})!}{r_{l+1}! \cdots   r_{n-1}!} \left(\frac{1}{k}\right)^{N_{l+1}}\nonumber\\
& \coloneqq W_{j,1}-W_{j,i,2} -(n-2)W_{j,3},
 \label{sum-W-j}
 \end{align}
 for $i=j+1,\ldots,n-1$.
 In (\ref{V2-eq2}) one can replace $r_i$ by, for instance, $r_{n-1}$ since $i\in\{j+1,\ldots,n-1\}$
and the summand is symmetric w.r.t.  $r_{j+1}, \ldots, r_{n-1}$.
 In the following $W_{j,2}\coloneqq W_{j,n-1,2}$.
 Finally, substituting (\ref{sum-W-j}) into (\ref{V2-eq3}) yields
\begin{equation}
\label{drift-V2}
k\times \Delta_{V_2}(\vect x_j')=W_j \sum_{i=j+1}^{n-1} x_i +\beta_j,
\end{equation}
with
\begin{equation}
\label{def-Tj-cross}
W_j\coloneqq(j+1){k-n+j+1\choose k-n}\left(W_{j,1}-W_{j,2} -(n-2)W_{j,3}\right)+n-2.
\end{equation}
Let us now turn to the calculation of $W_j$.
A glance at the definitions of  $W_{j,1}$,  $W_{j,2}$ and  $W_{j,3}$  in (\ref{sum-W-j}) shows that they are all  of the form (\ref{app:Gm0}) in Appendix \ref{app:generic}.
Specifically, 
\[
W_{j,1}=G_j(g_1;0,n),\quad W_{j,2}=G_j(g_2;-1,n),\quad W_{j,3}=G_j(g_1;-1,n),
\]
where the mappings $g_1$ and $g_2$ are defined in (\ref{def:mappings-g}).
Hence, we may apply 
Lemma \ref{lem:Gm} and use (\ref{F0:g1})-(\ref{F0:g2}) to calculate them. We find
\begin{equation}
\label{Wj1}
W_{j,1}= \frac{1}{k}\sum_{l=1}^j {j\choose l}l (-1)^{l+1}  F_0(n-l-1,1)= k\sum_{l=1}^j{ j\choose l} \frac{l (-1)^{l+1}}{(k-n+l+1)^2},
\end{equation}
by using (\ref{F0:g1}), 
\begin{equation}
\label{Wj2}
W_{j,2}= \frac{1}{k}\sum_{l=1}^j {j\choose l}l (-1)^{l+1}  F_0(n-l-1,0)= \sum_{l=1}^j{ j\choose l} \frac{l (-1)^{l+1}}{(k-n+l+1)^2},
\end{equation}
by using (\ref{F0:g2}),  and 
\begin{equation}
\label{Wj2}
W_{j,3}= \frac{1}{k}\sum_{l=1}^j {j\choose l}l (-1)^{l+1}  F_0(n-l-1,0)= \sum_{l=1}^j{ j\choose l} \frac{l (-1)^{l+1}}{k-n+l+1},
\end{equation}
by using (\ref{F0:g1}).
Therefore,
\[
W_{j,1}-W_{j,2} -(n-2)W_{j,3}= \sum_{l=1}^j {j\choose l} l(-1)^{l+1}  \frac{k-1 -(n-2)(k-n+l+1)}{(k-n+l+1)^2},
\]
and finally,
\begin{align}
\label{Wj}
W_j&=(j+1){k-n+j+1\choose k-n}  \sum_{l=1}^j {j\choose l} l(-1)^{l+1}  \frac{k-1 -(n-2)(k-n+l+1)}{(k-n+l+1)^2}
+n-2.
\end{align}
\begin{lemma}
\label{lem:W}
\begin{equation}
\label{lem:W1}
W_1=\frac{(k-1)(k-n+1)}{k-n+2}-(n-2)(k-n) 
\end{equation}
and 
\begin{equation}
\label{eq:diff-Wj}
W_{j+1}-W_j=\frac{(k-1)(k-n+1)}{k-n+j+2},
\end{equation}
yielding
\begin{equation}
W_j=(k-1)(k-n+1)\sum_{i=2}^{j+1} \frac{1}{k-n+i}-(n-2)(k-n),
\label{lem:Wj}  \quad j=1,\ldots,n-2.
\end{equation}
\end{lemma}
The proof of Lemma \ref{lem:W} is given in Appendix \ref{app:W}.

We can now conclude the proof of Proposition \ref{prop:drift}: multiplying (\ref{drift-V2}) by $b$,  where $W_j$ is given in (\ref{lem:Wj}),
and adding it to (\ref{Delta-V1-foster2}) gives (\ref{prop:delta-2}), with $\delta_j:=\gamma_j+b\beta_j$.  It is shown in Appendix \ref{app:finitess}
that $\gamma_j$ (defined in (\ref{def-gamma-j})) and $\beta_j$ (defined in (\ref{def-beta-j})) are finite for $j=1,\ldots,n-2$. 
\end{proof}
\subsection{Proof of Lemma \ref{lem:T}}
\label{app:Tj}
\begin{proof}
We start from $\Gamma_j$ given  in (\ref{Gamma-j-2})  which, from (\ref{app:Gm0}) and Lemma \ref{lem:Gm}, is given by
\begin{align}
\Gamma_j&=G_j(g_3;-1,n)\nonumber\\
&=\frac{1}{k}\sum_{l=1}^j {j\choose l}(-1)^{l+1}l F_0(g_3;n-l-1,0)\nonumber\\
&= \sum_{l=1}^j {j\choose l}(-1)^{l}l\frac{k-n+l}{(k-n+l+1)^2},
\label{value-Gamma-j-3}
\end{align}
by using (\ref{F0:g3}). 
Write $T_j(k,n)$ for $T_j$ to make explicit the dependence on $k$ and $n$. Then, from (\ref{def:Tj}),
\begin{align}
T_{j+1}(k,n)-T_j(k,n)&={k-n+j+1\choose k-n}  \left[(k-n+j+2)\Gamma_{j+1}-(j+1)\Gamma_j\right]\nonumber\\
&={k-n+j+1\choose k-n} \Biggl[(k-n+j+2) \sum_{l=1}^{j+1}{j+1\choose l} (-1)^l l \frac{k-n+l}{(k-n+l+1)^2}\nonumber\\
&\qquad
-(j+1)\sum_{l=1}^j {j\choose l} (-1)^l l \frac{k-n+l}{(k-n+l+1)^2}\Biggr] \quad \hbox{via } (\ref{value-Gamma-j-3})\nonumber\\
&={k-n+j+1\choose k-n} \Biggl[\sum_{l=1}^j (-1)^l l \frac{k-n+l}{(k-n+l+1)^2}\Bigg( (k-n+j+2) {j+1\choose l}
-  (j+1){j\choose l}\Bigg)\nonumber\\&\qquad
+ (-1)^{j+1} \frac{k-n+j+1}{k-n+j+2}\Biggr] \nonumber\\
&={k-n+j+1\choose k-n}\Biggl[\sum_{l=1}^j {j+1 \choose l} (-1)^l l \frac{k-n+l}{k-n+l+1}
+ (-1)^{j+1}(j+1) \frac{k-n+j+1}{k-n+j+2}\Biggr]\nonumber\\
&= {k-n+j+1\choose k-n} \sum_{l=1}^{j+1} {j+1 \choose l} (-1)^l l \,\frac{k-n+l}{k-n+l+1}\nonumber\\
&={k-n+j+1\choose k-n} \left[ \sum_{l=1}^{j+1} {j+1 \choose l} (-1)^l l  - \sum_{l=1}^{j+1} {j+1 \choose l} (-1)^l \frac{l}{k-n+l+1}\right].
\nonumber
\end{align}
Using (\ref{poly-result}),
\begin{equation}
\label{diff-Tj}
T_{j+1}(k,n)-T_j(k,n)= {k-n+j+1\choose k-n}\sum_{l=1}^{j+1} {j+1 \choose l} \frac{ (-1)^{l+1}l}{k-n+l+1}.
\end{equation}
We claim that 
\begin{equation}
\label{claim}
T_{j+1}(k,n)-T_j(k,n)=\frac{k-n+1}{k-n+j+2},
\end{equation}
for all $j=1,\ldots,n-2$ and $n\geq 3$.
When $j=1$,  (\ref{diff-Tj}) yields
\begin{align*}
T_2(k,n)-T_1(k,n)&=\frac{(k-n+1)(k-n+2)}{2}\left(\frac{2}{k-n+2} - \frac{2}{k-n+3}\right)
=\frac{k-n+1}{k-n+3},
\end{align*}
which shows that (\ref{claim}) is true for $j=1$ and for all $n\geq 3$.
Now, assume that $T_{j+1}(k,n)-T_j(k,n)=\frac{k-n+1}{k-n+j+2}$ for $j=1,\ldots, J-1$  and $n\geq 3$, and let us show that  $T_{J+1}(k,n)-T_{J}(k,n)=\frac{k-n+1}{k-n+J+2}$. 
From (\ref{diff-Tj}) and (\ref{basic-id}) in that order, we obtain
\begin{align*}
T&_{J+1}(k,n)-T_J(k,n)={k-n+J+1\choose k-n}\sum_{l=1}^{J} {J+1 \choose l}  \frac{(-1)^{l+1}l}{k-n+l+1}
+ {k-n+J+1\choose k-n}\frac{ (-1)^{J+2}(J+1)}{k-n+J+2}\\
&=
 {k-n+J+1\choose k-n}  \left(\frac{J+1}{J}\right) \sum_{l=1}^{J} \frac{(-1)^{l+1}}{k-n+l+1}\left[ {J\choose l}l+ {J\choose l-1}(l-1)\right]
+ {k-n+J+1\choose k-n} \frac{(-1)^{J}(J+1)}{k-n+J+2}\\
&=\frac{k-n+J+1}{J} \left(T_J(k,n)-T_{J-1}(k,n)\right)
+ {k-n+J+1\choose k-n}\left(\frac{J+1}{J}\right)\sum_{l=2}^{J} {J\choose l-1}\frac{ (-1)^{l+1} (l-1)}{k-n+l+1}\\&\qquad 
+ {k-n+J+1\choose k-n} \frac{(-1)^{J}(J+1)}{k-n+J+2}\\
&=\frac{k-n+J+1}{J} \left(T_J(k,n)-T_{J-1}(k,n)\right)
+ {k-n+J+1\choose k-n}\left(\frac{J+1}{J}\right)\sum_{l=1}^{J} {J\choose l}\frac{ (-1)^{l} l}{k-(n-1)+l+1}\\
&=\frac{k-n+J+1}{J} \left(T_J(k,n)-T_{J-1}(k,n)\right)
-\frac{k-n+1}{J} \left(T_J(k,n-1)-T_{J-1}(k,n-1)\right)\\
&=\frac{k-n+J+1}{J}\left(\frac{k-n+1}{k-n+J+1}\right)
-\frac{k-n+1}{J} \left(\frac{k-n+2}{k-n+J+2}\right)=\frac{k-n+1}{k-n+J+2},
\end{align*}
by using the induction hypothesis. This concludes the induction step and proves the validity of (\ref{claim}).

Letting $j=1$ in (\ref{value-Gamma-j-3}) and using (\ref{def:Tj}), we obtain
\begin{equation}
\label{app:T1}
T_1(k,n)=-2{k-n+2\choose k-n}\frac{k-n+1}{(k-n+2)^2}+1=-\frac{(k-n)^2 +k-n-1}{k-n+2}.
\end{equation}
The recursion (\ref{claim}) together with (\ref{app:T1}) yields
\begin{align*}
T_j(k,n)&=-\frac{(k-n)^2 +k-n-1}{k-n+2}  +(k-n+1)\sum_{i=3}^{j+1}\frac{1}{k-n+i}\nonumber\\
&=-(k-n) +(k-n+1)\sum_{i=2}^{j+1}\frac{1}{k-n+i}.
\end{align*}
This completes the proof of Lemma \ref{lem:T}.
\end{proof}
\subsection{Proof of Lemma \ref{lem:W}}
\label{app:W}
\begin{proof}
Letting $j=1$ in (\ref{Wj}) yields (\ref{lem:W1}). Assume that $j=2,\ldots,{n-2}$. We have from (\ref{Wj})
\[
W_{j+1}-W_j=(k-1)A_j -(n-2)B_j,
\]
with
\begin{align*}
A_j&=\frac{(k-n+j+2)!}{(k-n)! (j+1)!} \sum_{l=1}^{j+1}  {j+1\choose l} \frac{(-1)^{l+1} l }{(k-n+l+1)^2}
-\frac{(k-n+j+1)!}{(k-n)! j!} \sum_{l=1}^{j}  {j\choose l} \frac{(-1)^{l+1} l }{(k-n+l+1)^2},
\end{align*}
and
\begin{align*}
B_j&=\frac{(k-n+j+2)!}{(k-n)! (j+1)!} \sum_{l=1}^{j+1}  {j+1\choose l}  \frac{(-1)^{l+1} l}{k-n+l+1}
-\frac{(k-n+j+1)!}{(k-n)! j!} \sum_{l=1}^{j}  {j\choose l}  \frac{(-1)^{l+1} l}{k-n+l+1}.
\end{align*}
Let us show that $B_j=0$. By considering separately the term corresponding to $l=j+1$ in the first sum in $B_j$, we have
\begin{align*}
B_j&={k-n+j+1\choose k-n}     (-1)^{j+2}(j+1)
+\frac{(k-n+j+1)!}{(k-n)!}\sum_{l=1}^j \frac{(-1)^{l+1} l}{(j-l)!l! (k-n+l+1)}\left[\frac{k-n+j+2}{j+1-l}-1\right]\\
 &={k-n+j+1\choose k-n}     (-1)^{j+2}(j+1) +{k-n+j+1\choose k-n} \sum_{l=1}^j {j+1\choose l}(-1)^{l+1}l \\
 &=-{k-n+j+1\choose k-n} \sum_{l=1}^{j+1} {j+1\choose l}(-1)^l l\\
 &=0,
 \end{align*}
by using (\ref{poly-result}). Hence,
\[
W_{j+1}-W_j=(k-1)A_j.
\]
It remains to show that $A_j=\frac{k-n+1}{k-n+j+2}$, which will prove (\ref{eq:diff-Wj}).
By mimicking the calculation of $B_j$ we find
\[
A_j ={k-n+j+1\choose k-n}\sum_{l=1}^{j+1} {j+1\choose l}\frac{(-1)^{l+1} l}{k-n+l+1}=\frac{k-n+1}{k-n+j+2},
\]
where the second equality has been proven in Appendix \ref{app:Tj} (see  (\ref{diff-Tj}) and (\ref{claim})). This completes the proof.
\end{proof}
\subsection{Finiteness of $\gamma_j$ and $\beta_j$}
\label{app:finitess}
\begin{proof}
Recall the definition of $\gamma_j$ in (\ref{def-gamma-j}). We get
\begin{align}
|\gamma_j|&\leq j(j+1) {k-n+j+1\choose k-n} \sum_{r_1\geq 0,\ldots,r_{n-1}\geq 0} 
\frac{(|\vect{r}|)!}{r_1!\cdots  r_{n-1}!} \left(\frac{1}{k}\right)^{|\vect{r}|}
\left(\sum_{s=1}^{n-1}r^2_s +2\sum_{s=1}^{n-1}r_s+n-1\right)+n-1+2j\nonumber\\
&=j(j+1) {k-n+j+1\choose k-n} \Biggl[\sum_{s=1}^{n-1}\sum_{\vect{r}\geq 0} 
\frac{(|\vect{r}|)!}{r_1!\cdots  r_{n-1}!} \left(\frac{1}{k}\right)^{|\vect{r}|}(r^2_s +2r_s)
+(n-1)\sum_{\vect{r}\geq 0} 
\frac{(|\vect{r}|)!}{r_1!\cdots  r_{n-1}!} \left(\frac{1}{k}\right)^{|\vect{r}|}\Biggr]\nonumber\\&\quad+n-1+2j.
\label{app:B:eq1}
\end{align}
Since
\begin{align*}
\sum\limits_{r_1\geq 0,\ldots,r_{n-1}\geq 0} 
\frac{(r_1+\cdots+r_{n-1})!}{r_1!\cdots  r_{n-1}!} \left(\frac{1}{k}\right)^{\sum_{l=1}^{n-1} r_l}(r^2_s +2r_s)
\end{align*}
has the same value for $s=1,\ldots,n-1$, we may
replace $r^2_s+2r_s$ by $r_{n-1}^2+2r_{n-1}$. Using  (\ref{gen-formula}), (\ref{app:diff}), and (\ref{app:diff2})  in Appendix \ref{app:useful-formulas}, we may rewrite
(\ref{app:B:eq1}) as follows:
\begin{align*}
|\gamma_j|&\leq
(n-1)j(j+1) {k-n+j+1\choose k-n} \Bigg(\frac{k}{k-n+1} 
+\frac{2k}{(k-n+1)^2} + \frac{(k-n+3)k}{(k-n+1)^3}\Bigg) +n-1+2j,
\end{align*}
which is finite for $j=1,\ldots,n-2$ and  $k\geq n$.

Next, we address the finiteness of $\beta_j$, defined in  (\ref{def-beta-j}). To simplify notation, define $\rho_j\coloneqq j(j+1){k-n+j+1\choose k-n}$.
Then,
\begin{align}
|\beta_j|&\leq  
\rho_j\sum_{\vect{r}\geq 0}  \frac{(|\vect{r}|)!}{r_1! \cdots r_{n-1}!}
 \left(\frac{1}{k}\right)^{|\vect{r}|}\left(|\alpha_j(\vect r)| + \frac{j^2-j}{2}\right)+(n-2)j\nonumber\\
 &=\rho_j\sum_{\vect{r}\geq 0}  \frac{(|\vect{r}|)!}{r_1! \cdots r_{n-1}!}
 \left(\frac{1}{k}\right)^{|\vect{r}|}|\alpha_j(\vect r)| +\frac{\rho_j k}{k-n+1}\left(\frac{j^2-j}{2}\right) +(n-2)j,
\label{bound-beta-j}
 \end{align}
by using (\ref{gen-formula}) with $J=n-1$, $L=0$ and $z_i=1/k$. We are left with  proving that  the sum in (\ref{bound-beta-j}) is finite. By noting that $\alpha_j(\vect r)$, defined in
(\ref{def-alpha-j}), is a finite sum composed of the terms $r_i$ and $r_i r_l$, $i\not=l$, we conclude from (\ref{app:diff}) and (\ref{app:diff-cross}) applied with $J=n-1$, 
$L=0$, and $z_i=1/k$ that this sum is finite. This proves the finiteness of $\beta_j$.
\end{proof}
\section{Proof of Proposition 8.1}
\label{app:expYproof}
\begin{proof}
The proof uses Theorem 1 in \cite{Tweedie1983}. We will apply this theorem with the {\em  finite}  set $A_M=\{\vect x\in S: |\vect x|\leq M\}$ ($M<\infty$ will be defined later on)
and the Lyapunov function $V(\vect x)$ used in the proof of Proposition \ref{prop:drift} (cf. (\ref{def:V})). To this end, we need to check that the following conditions are fulfilled: there exist constants $M$, $c$, and $d$ such that
\begin{align}
\label{hyp}
\sum_{\vect y\in A^c_M} q(\vect x, \vect y) V(\vect y) -  V(\vect x)&\leq - c |\vect x| -d, \quad \forall \vect x\in A^c_M, \\
c|\vect x|+d&\geq 0,  \quad \forall \vect x\in A^c_M,
\label{hyp2}\\
V(\vect x)&\geq c|\vect x|+d, \quad \forall \vect x\in A^c_M,
\label{hyp3}
\end{align}
and 
\begin{align}
\label{hyp4}
\sup_{\vect x\in A_M}\sum_{\vect y\in A^c_M} q(\vect x, \vect y)V(\vect y)<\infty.
\end{align}
If these conditions hold, 
then\footnote{Theorem 1 in \cite{Tweedie1983} requires that $f(x)\geq 0$ for all $x$ in the  state-space. A glance at the proof of this
theorem shows that it is still true if $f(x)\geq 0$ for all $x\in A^c_M$ and $\int_{A_M} f(x)\pi(dx)$ is finite, thereby justifying (\ref{hyp2}).}  Theorem 1 in \cite{Tweedie1983} will imply that $\E[c|\vect{Y}|+d]<\infty$,  which in turn implies that $\E[|\vect{Y}|]<\infty$.

We have shown in Proposition \ref{prop:drift} that if
$b=-\frac{2(1-\alpha)}{n-2}$ with $0<\alpha<\frac{1}{n-1}$
 (and by setting $m=k-n\geq 1$), then for $j=1,\ldots,n-2$,
\begin{align}
\Delta_V(\vect x)=\left\{\begin{array}{ll}
\frac{-2m\alpha }{k} |\vect x|+\frac{n-1}{k}(1+\alpha (m+1))), \qquad \mbox{for $\vect x\in S-S^\star$,}\\
-\frac{2}{k}\left(\frac{(m+1)^2-\alpha(m^2 +mn+n-1)}{n-2}\sum_{i=2}^{j+1} \frac{1}{m+i}+m \alpha \right)|\vect x|+ \frac{\xi_j}{k},  \qquad\mbox{for $\vect x\in S^\star_j$,}
                                             \end{array} 
                                      \right.
\label{exp:drift}
\end{align}
where
\[
\xi_j:=2j \left( \frac{(m+1)^2-\alpha(m^2 +mn+n-1)}{n-2}\sum_{i=2}^{j+1} \frac{1}{m+i}+m\alpha \right)+\delta_j.
\]
In this case we know from (\ref{def:V-2}) that $V(\vect x)\geq 0$ for all $\vect{x}\in S$, which implies
that 
\begin{equation}
\label{prop:inq0}
\sum_{\vect y\in A_M^c} q(\vect x, \vect y) V(\vect y) -  V(\vect x)  \leq \sum_{\vect y\in S} q(\vect x, \vect y) V(\vect y) -  V(\vect x)  =\Delta_V(\vect x), \quad \forall \vect{x}\in S.
\end{equation}
The finiteness of the set $A_M$, the finiteness of  $V(\vect x)$ for any $\vect{x}\in S$ (see (\ref{exp:drift})), and (\ref{prop:inq0}) show that (\ref{hyp4}) is satisfied. 
On the other hand,  (\ref{exp:drift}) and (\ref{prop:inq0}) show that  (\ref{hyp}) holds if (for instance)
\begin{align}
c&= \frac{2}{k}\min\Bigg\{ m\alpha , \frac{(m+1)^2-\alpha(m^2 +mn+n-1)}{n-2}  \sum_{i=2}^{j+1} \frac{1}{m+i}+m \alpha,~ j=1,\ldots,n-2 \Bigg\}\label{def-c0}\\
d&= -\frac{1}{k}\max\left\{(n-1)(1+\alpha (m+1)), \xi_1,\ldots,\xi_{n-2}\right\}.\label{def-d}
\end{align}
We have already observed  in (\ref{proof:inq}) that $(m+1)^2-\alpha(m^2 +mn+n-1)>0$  for $m\geq 1$ and $\alpha<\frac{1}{n-1}$. Hence, 
\begin{equation}
c=\frac{2m\alpha}{k}=\frac{2(k-n)\alpha}{k}>0,
\label{def-c}
\end{equation} 
when $0<\alpha <\frac{1}{n-1}$. 

We are left with proving that one can find $M$ such that (\ref{hyp2}) and (\ref{hyp3}) hold for $c$ and $d$ given above.
The finiteness of $\delta_j$ (see Proposition \ref{prop:drift}) yields
the finiteness of $\xi_j$, which in turn yields the finiteness of $d$. Since $c>0$ and $d$ is finite, there exits $M_1$ such that $c|\vect x|+d \geq 0$
 for all $\vect x\in A^c_M$ with $M>M_1$. This shows that (\ref{hyp2}) is true for any $M>M_1$.
It remains to show that condition (\ref{hyp3}) holds.  We have
\begin{equation}
\label{exp:eq1}
V(\vect x)-c|\vect x|-d= \sum_{i=1}^{n-1} x_i^2 -\frac{2(1-\alpha)}{n-2}\sum_{1\leq i<l\leq n-1} x_i x_l - \frac{2(k-n)\alpha}{k}|\vect x|-d.
\end{equation}
Differentiating (\ref{exp:eq1}) w.r.t. $x_i$ and equating the result to zero, gives (Hint: $\frac{d}{dx_i}\sum_{1\leq i<l\leq n-1} x_i x_l =|\vect x|-x_i$ for $i=1,\ldots,n-1$)
\begin{equation}
\label{exp:eq2}
x_i(n-1-\alpha)-(1-\alpha)|\vect x| =\frac{(n-2)(k-n)\alpha}{k}, \quad i=1,\ldots,n-1.
\end{equation}
Summing up both sides of (\ref{exp:eq2}) for $i=1,\ldots,n-1$ yields
\begin{equation}
\label{exp:eq3}
|\vect x|=\frac{(k-n)(n-1)}{k},
\end{equation}
which, with the help of (\ref{exp:eq2}), yields
\begin{equation}
\label{exp:eq4}
x_i=\frac{k-n}{k}.
\end{equation}
This shows that $V(\vect x)-c|\vect x|-d$ has a unique local extremum, at point $\vect x_0=\left(\frac{k-n}{k}, \ldots, \frac{k-n}{k}\right)$.
The Hessian matrix 
\begin{align*}{\bf M}\coloneqq \left[\frac{\partial}{\partial x_i x_j} (V(\vect x)-c|\vect x|-d)\right]_{i,j}\end{align*}
is given by 
${\bf M}=2 {\bf N}$, where the matrix ${\bf N}$ has all its entries equal to  $-\frac{1-\alpha}{n-2}$ except the diagonal entries that are all equal to $1$.
Let us show that {\bf M}, or equivalently {\bf N}, is a positive definite matrix, which will ensure that $\vect x_0$ is a global minimum.
For any $\vect v=(v_1~\ldots~v_{n-1})\not=0$, we have
\begin{equation}
\label{hessian}
(n-2)\vect v\, {\bf N}\,\vect v^T=(n-1-\alpha)\sum_{i=1}^{n-1} v^2_i-(1-\alpha)|\vect v|^2.
\end{equation}
Take $\alpha=0$. Then, using a proof by induction on $n$, we obtain
\begin{equation}
\label{hessian2}
(n-2)\vect v\, {\bf N}\,\vect v^T=\sum_{1\leq i<l\leq n-1} (v_i-v_l)^2>0,
\end{equation}
for all $\vect v\not=0$. Because the mapping $\alpha \to \vect v\, {\bf N}\,\vect v^T$ in (\ref{hessian}) is continuous, (\ref{hessian2}) shows that there exists
$\alpha_0>0$ such that $\vect v{\bf N}\,\vect v^T>0$ for all $\alpha\in [0,\alpha_0)$, $\vect v\not=0$.
This shows that $V(\vect x)-c|\vect x|-d$ has a global minimum, located at 
$\vect x_0=\left(\frac{k-n}{k}, \ldots, \frac{k-n}{k}\right)$.
Hence, 
\begin{equation}
V(\vect x)-c|\vect x|-d\geq V(\vect x_0)-c|\vect x_0|-d =-\alpha (n-1)\left(\frac{k-n}{k}\right)^2 - d,
\end{equation}
after elementary algebra. From (\ref{def-d}) we get 
\[
-d\geq \frac{(n-1)(1+\alpha(m+1))}{k},
\]
which yields
\begin{equation}
V(\vect x_0)-c|\vect x_0|-d\geq -\alpha (n-1)\left(\frac{k-n}{k}\right)^2 +\frac{(n-1)(1+\alpha (m+1))}{k}.
\label{exp:eq3}
\end{equation}
When $\alpha =0$ the r.h.s. of (\ref{exp:eq3}) is strictly positive. Since it is a continuous function of $\alpha$,  there exists $\alpha_0>0$, such 
that $V(\vect x_0)-c|\vect x_0|-d>0$ for all $\alpha \in (0,\min(\alpha_0,1/(n-1))$.  This proves that
$V(\vect x)\geq c|\vect x|+d$ for all $\alpha \in (0,\min(\alpha_0,1/(n-1))$.

In summary,  conditions (\ref{hyp})-(\ref{hyp4}) are satisfied for  $b=-\frac{2(1-\alpha)}{n-2}$ with $\alpha \in (0,\min(\alpha_0,1/(n-1))$ and any $M>\max (M_1,n-1)$ 
(we need $M>n-1$ so that $(1,\ldots,1)\in A_M$). This concludes the proof.
\end{proof}

\end{document}